\documentclass[runningheads]{llncs}
\def \VersionWithComments{}
\usepackage[T1]{fontenc}
\usepackage[right]{lineno}
\usepackage{fullpage}

\usepackage{packages}
\usepackage{macros}
\usepackage{alltt}     
\begin{document}
\title{Symbolic Analysis and Parameter Synthesis for Time
  Petri Nets Using Maude and SMT Solving\thanks{Supported by CNRS
    INS2I project ESPRiTS and
    PHC Aurora AESIR.}}

\titlerunning{Analyzing Parametric Time Petri Nets Using Maude with SMT}
\author{Jaime Arias\inst{1} \and
  Kyungmin Bae\inst{2}
\and
  Carlos Olarte\inst{1} \and
  Peter Csaba {\"O}lveczky\inst{3} \and
  Laure Petrucci\inst{1} \and
  Fredrik R{\o}mming\inst{4} }

\authorrunning{J. Arias et al.}
\institute{LIPN, CNRS UMR 7030, Université Sorbonne Paris Nord, France
\and
Pohang University of Science and Technology, South Korea \and
University of Oslo, Norway \and
University of Cambridge, UK}  \maketitle              \begin{abstract}
    Parametric time Petri nets with inhibitor arcs (PITPNs) support flexibility for
 timed systems
 by allowing
 parameters in firing bounds.
In this paper we present  and prove correct a concrete and a symbolic rewriting
logic semantics for PITPNs. We show how this allows us to use Maude combined with SMT solving to
provide sound and complete formal analyses for PITPNs.
We develop a new general folding
 approach for symbolic reachability that
terminates whenever the parametric state-class graph of the PITPN is finite.
We explain how almost all formal analysis and parameter synthesis
supported by the state-of-the-art PITPN tool 
\romeo{} can be done in Maude with SMT.  In addition,  we also support
 analysis and parameter 
synthesis from \emph{parametric} initial markings, as well as  full LTL model
checking and analysis  with 
user-defined execution strategies.
Experiments on three benchmarks show that our methods outperform
\romeo{} in many cases.


 \end{abstract}

\section{Introduction}\label{sec:intro}

Time(d) Petri nets
\cite{Merlin74,DBLP:reference/crc/VernadatB07,DBLP:books/daglib/0023756}
have been extensively 
used to model real-time systems. In time Petri nets, firing conditions are given as time
intervals within which an enabled transition must fire.
However, in system design we often do not  know in advance the concrete
values of key system parameters, and want to find 
those  values  that make the system behave as desired. \emph{Parametric
  time Petri nets with inhibitor arcs}
(PITPNs)~\cite{paris-paper,DBLP:conf/formats/GrabiecTJLR10,EAGPLP13,DBLP:journals/fuin/LimeRS21}
extend time Petri
nets to the setting  where  bounds on
when transitions can fire are unknown or only partially
known. 

The
modeling and formal analysis of PITPNs---including synthesizing the
values of the parameters which make the system satisfy desired
properties---are supported by the state-of-the-art tool
\romeo~\cite{romeo}, which has been applied to a number of
applications,  including oscillatory biological
systems~\cite{romeo-biology}, aerial video tracking
systems~\cite{romeo-aerial}, and distributed software
commissioning~\cite{romeo-software}.
\romeo{}
supports the analysis and parameter synthesis for reachability (is a
certain marking reachable?), liveness (will a certain marking be
reached in all behaviors?), time-bounded ``until,'' and bounded
response (will each $P$-marking 
be followed by a $Q$-marking within time $\Delta$?), all from \emph{concrete}
initial markings. 
\romeo{} does  not support a number of desired features,
including:
\begin{itemize}
  \item Broader set of system properties, e.g.,  full (i.e.,
    nested) temporal logic.
\item  Start with \emph{parametric} initial
  markings and synthesize also the initial markings that make the
  system satisfy desired properties.
\item Analysis with user-defined execution strategies. For
      example, what  happens if I always choose to fire
      transition $t$ instead of $t'$ when they are both enabled at the
      same time? 
It is often possible to  \emph{manually
      change the model} to analyze the system under such scenarios, but
      this is arduous and error-prone. 
  \item Providing a ``testbed'' 
  for PITPNs in which different analysis methods and algorithms can quickly be
     developed, tested, and evaluated.  
     This is not well
    supported by \romeo, which is a high-performance tool with
    dedicated algorithms implemented in \texttt{C++}.
\end{itemize}

   PITPNs  do not support many features needed for large
   distributed systems, such as user-defined data types and functions,
   as in, e.g.,  colored Petri 
   nets~\cite{DBLP:books/daglib/0023756}.

   Rewriting
   logic~\cite{Mes92,20-years}---supported by the Maude language and
    tool~\cite{maude-book}, and by Real-Time
   Maude~\cite{tacas08,wrla14} for real-time systems---is an
   expressive logic for 
   distributed and real-time systems. In rewriting logic, any
   computable data type can be specified as an (algebraic)  equational
specification, and  the dynamic behaviors of a system are
specified by rewriting rules over terms (representing states). Because
of its expressiveness, Real-Time Maude has been successfully applied
to a number of large and sophisticated real-time systems---including 50-page
active networks and IETF protocols~\cite{aer-journ,sefm09},
state-of-the-art wireless sensor network algorithms involving areas,
angles, etc.~\cite{wsn-tcs}, scheduling algorithms with unbounded
queues~\cite{fase06}, airplane turning
algorithms~\cite{airplane-journ}, and so on---beyond the scope of
most popular formalisms for real-time systems. Its
expressiveness has also made Real-Time Maude a useful semantic framework
and formal analysis backend for (subsets of) industrial
modeling
languages~\cite{fmoods10,musab-l,ptolemy-journ,carolyn}. 

This expressiveness comes at a price: most analysis problems
are undecidable in general. Real-Time Maude uses
explicit-state analysis  where  only
\emph{some} points in time are visited.  All possible system
behaviors are therefore \emph{not} analyzed (for dense time domains),
and hence the  analysis is  unsound 
in many cases~\cite{wrla06}.

This paper exploits the recent integration of SMT solving into Maude
  to address the
first problem above (more features for PITPNs)  and to take the second step
towards addressing the second problem (developing sound and complete
analysis methods for rewriting-logic-based real-time systems).

Maude combined with SMT solving, e.g., as implemented in the Maude-SE
tool~\cite{maude-se}, allows us to perform \emph{symbolic rewriting}
of ``states'' $\phi\; ||\; t$, where the term $t$ is a state pattern
that
contains variables, and $\phi$ is an SMT constraint restricting
the possible values of those variables.

Section~\ref{sec:concrete}  provides a ``concrete'' rewriting
logic semantics for (instantiated) PITPNs
in ``Real-Time Maude style''~\cite{rtm-journ}. In a dense-time setting, 
such as for PITPNs, this model is not executable. 
Section~\ref{sec:concrete-ex} shows how we can do(in general unsound) time-sampling-based analysis where time
increases in discrete steps, of concrete nets, to quickly experiment
with different  values for the parameter. 

Section~\ref{sec:sym} gives a Maude-with-SMT semantics for 
\emph{parametric} PITPNs,  and shows how
to perform (sound) symbolic analysis of such nets using Maude with
SMT. However, existing symbolic reachability analysis methods,
including  ``folding'' of symbolic states, may fail to terminate even
when the state class graph of the PITPN is finite (and
hence \romeo{} analysis terminates). We therefore develop and
implement a new  method for ``folding'' symbolic states for
reachability analysis in Maude-with-SMT, and show that this new
reachability analysis method terminates whenever the state class graph
of the PITPN is finite.

In Sections~\ref{sec:sym} and \ref{sec:analysis} we show how a range of 
 formal analyses and parameter synthesis can be performed with
Maude-with-SMT, including unbounded and time-bounded reachability
analysis. We  show in Section~\ref{sec:analysis} how all analysis
methods supported by \romeo---with one small exception: 
the time bounds in some temporal formulas cannot be parameters---also
can be performed in Maude-with-SMT. In addition, we support state
properties on both markings and ``transition clocks,'' analysis and
parameter synthesis for \emph{parametric} initial markings, model
checking full (i.e., nested) temporal logic formulas, and
analysis w.r.t.\ user-defined execution strategies, as illustrated in
Section~\ref{sec:analysis}.
Our methods are formalized/implemented in Maude itself, using Maude's
meta-programming features. This makes it very easy to
develop new analysis methods for PITPNs.

This work also constitutes the second step in our quest to develop
sound and complete formal analysis methods for dense-time real-time
systems in Real-Time Maude. One  reason
for presenting both a ``standard'' Real-Time Maude-style concrete
semantics in Section~\ref{sec:concrete} \emph{and} the symbolic
semantics in Section~\ref{sec:sym} is to explore how we can transform
Real-Time Maude models into Maude-with-SMT models for symbolic
analysis.
In our first step in this quest, we studied symbolic rewrite methods
for the much simpler parametric timed automata (PTA)~\cite{ftscs22}.
In~\cite{ftscs22} we
specify a new rewrite theory for \emph{each} automaton, whereas in
this paper we specify a single rewrite theory (``interpreter'') for
\emph{all} PITPNs. Furthermore, no equations or user-defined functions
are needed for PTAs, in  contrast to the models in this
paper. Finally, known folding methods are sufficient for PTAs, whereas
we had to develop stronger  folding methods for PITPNs.  

In Section~\ref{sec:benchmarks} we benchmark both \romeo{} and
our Maude-with-SMT methods on three PITPNs. Somewhat surprisingly, in
many cases our high-level prototype outperforms \romeo{}. 
We also discovered that \romeo{} answered ``maybe'' in some cases
where Maude found solutions, and that \romeo{} sometimes failed to
synthesize parameters even when solutions existed.

All  executable Maude files with 
analysis commands, tools for translating \romeo{} files into Maude,  and data 
from the benchmarking are available at~\cite{pitpn2maude}.

\section{Preliminaries} \label{sec:prelim}
This section introduces
transition systems, bisimulation \cite{model-checking},
parametric time Petri nets with inhibitor arcs \cite{paris-paper},
rewriting logic \cite{Mes92}, rewriting modulo SMT \cite{rocha-rewsmtjlamp-2017}
and Maude/Maude-SE \cite{maude-manual,maude-se}.

A \emph{transition system} $\mathcal{A}$ is a triple $(A,
a_{0}, \rightarrow_{\mathcal{A}})$,  where $A$ is a set of
\emph{states}, $a_{0}\in A$ is the \emph{initial state}, and
$\rightarrow_{\mathcal{A}}\,\subseteq A \times A$ is a \emph{transition
  relation}.
We say that $\mathcal{A}$ is \emph{finite} if the set of states
reachable by $\rightarrow_{\mathcal{A}}$ from $a_0$ is finite. 
A  relation
$\sim \,\subseteq  A \times B$ is a \emph{bisimulation}~\cite{model-checking}
from $\mathcal{A}$ to $\mathcal{B} = (B, b_0, \rightarrow_{\mathcal{B}})$
iff: (i) $a_0 \sim b_0$; and (ii) for all $a,b$ s.t. $a\sim b$:
if 
$a \rightarrow_{\mathcal{A}} a'$ then there is  a $b'$ s.t. $b
\rightarrow_{\mathcal{B}} b'$ and $a' \sim b'$, and,  vice versa,  if 
$b \rightarrow_{\mathcal{B}} b''$,
then there is a $a''$ s.t. $a \rightarrow_{\mathcal{A}} a''$ and
$a'' \sim b''$.

\subsection{Parametric Time Petri Nets with Inhibitor Arcs (PITPN).}\label{sec:PTPN}

We recall  the definitions from~\cite{paris-paper}. 
$\grandn$, $\grandqplus$, and $\grandrplus$ denote, resp., the 
natural numbers, the non-negative rational numbers,  and the non-negative real
numbers.
Throughout this paper, we assume a finite set
$\Param=\{\param_1, \dots,\param_{\ParamCard} \}$
of \emph{time parameters}. 
A \emph{parameter valuation} $\pi$ is a function
$\pi : \Param \rightarrow \grandrplus$.  
A (linear) \emph{inequality} over  $\Param$ is an expression
$\sum_{1 \leq i \leq \ParamCard} a_i \param_i \prec b$,
where $\prec \in \{<, \leq,=,\geq,>\}$
and $a_{i},b \in \grandr$.
A \emph{constraint} is a conjunction of such inequalities.
$\setP$ denotes the set of all constraints over $\Param$. 
A parameter valuation~$\pi$ \emph{satisfies} a constraint $K\in
\setP$, written $\pi \models K$,  
if the expression obtained by replacing each parameter~$\param$ in~$K$
with~$\pi(\param)$ evaluates to true. 
An interval $\interval$ of $\grandrplus$ is a $\grandqplus$-interval
if its left endpoint $\leftEP{\interval}$ belongs to $\grandqplus$ and
its right endpoint $\rightEP{\interval}$ belongs to $\grandqplus \cup
\{ \infty \}$. 
We denote by $\Interval(\grandqplus)$ the set of
$\grandqplus$-intervals. A parametric time interval is a function $\parInterval : {\grandqplus}^\Param
\rightarrow \Interval(\grandqplus)$ that associates with each parameter
valuation a $\grandqplus$-interval.
The set of parametric time intervals over $\Param$ is denoted  $\ParInterval(\Param)$.

\begin{definition}[PITPN]
    A \emph{parametric time Petri net with inhibitor arcs}  is
    a tuple 

    $\PN = \tuple{\Place, \Transition, \Param, \relPre{(.)},
      \relPost{(.)}, \relInhib{(.)}, \markingInit, \parIntervalStatic,
      \Kinit}$ where 
    \begin{itemize}
        \item $\Place = \{ \place_1, \dots, \place_\PlaceCard \}$ is a
          non-empty finite set (of \emph{places}), 
        \item $\Transition = \{ \transition_1, \dots,
          \transition_\TransitionCard \}$ is a non-empty finite set (of
          \emph{transitions}), with $P \cap T = \emptyset$, 
        \item $\Param = \{ \param_1, \dots, \param_\ParamCard \}$ is a
          finite set of \emph{parameters}, 
        \item $\relPre{(.)} \in
          [\Transition \rightarrow \grandn^\Place]$ is the 
          \emph{backward} \emph{incidence function}, 
       \item $\relPost{(.)} \in
          [\Transition \rightarrow \grandn^\Place]$ is the 
          \emph{forward} \emph{incidence function}, 
       \item $\relInhib{(.)} \in [\Transition\rightarrow\grandn^\Place]$
        is the
          \emph{inhibition function}, 
        \item $\markingInit \in \grandn^\Place$ is the \emph{initial marking},
        \item $\parIntervalStatic \in
          [\Transition\rightarrow\ParInterval(\Param)]$ assigns a \emph{parametric time interval} to each
          transition, and 
        \item $\Kinit \in \setP$ is the \emph{initial constraint} over $\Param$.
        \end{itemize}
If  $\Param = \emptyset$ then $\PN$ is a (non-parametric)
\emph{time Petri net with inhibitor arcs} (ITPN).    
\end{definition}

\begin{figure}[ht]
    \subfloat[][A PITPN~$\PN$.\label{fig:example:PITPN:param}]{
\begin{tikzpicture}[xscale=1.1, yscale=.8, >=stealth',bend angle=30,auto]

        \node [place1, tokens=1] at (1, 2) (A) [label=left:A] {};
        \node [place2, tokens=1] at (3, 2) (B) [label=left:B] {};
        \node [transition] (t1) at (0, 1) [label=left:$\transition_1 {[ \param_1^- , \param_1^+ ]} $] {};
        \node [transition] (t2) at (2, 1) [label=left:$\transition_2 {[ \param_2^- , \param_2^+ ]} $] {};
        \node [transition] (t3) at (4, 1) [label=left:$\transition_3 {[ \param_3^- , \param_3^+ ]} $] {};
        \node [place3] at (0, 0) (C) [label=left:C] {};
        \node [place4] at (2, 0) (D) [label=left:D] {};
        \node [place5] at (4, 0) (E) [label=left:E] {};

        \path
            (A) edge [post] node {} (t1)
            (t1) edge [post] node {} (C)
            (A) edge [post,-o] node {} (t2)
            (t2) edge [post] node {} (D)
            (B) edge [post] node {} (t2)
            (B) edge [post] node {} (t3)
            (t3) edge [post] node {} (E)
        ;

        \end{tikzpicture}
}
    \hfill
    \subfloat[][The ITPN~$\valuate{\PN}{\py}$.\label{fig:example:PITPN:valuated}]{
\begin{tikzpicture}[xscale=.9, yscale=.8, >=stealth',bend angle=30,auto]

        \node [place1, tokens=1] at (1, 2) (A) [label=left:A] {};
        \node [place2, tokens=1] at (3, 2) (B) [label=left:B] {};
        \node [transition] (t1) at (0, 1) [label=left:$\transition_1 {[ 5 , 6 ]} $] {};
        \node [transition] (t2) at (2, 1) [label=left:$\transition_2 {[ 3 , 4 ]} $] {};
        \node [transition] (t3) at (4, 1) [label=left:$\transition_3 {[ 1 , 2 ]} $] {};
        \node [place3] at (0, 0) (C) [label=left:C] {};
        \node [place4] at (2, 0) (D) [label=left:D] {};
        \node [place5] at (4, 0) (E) [label=left:E] {};

        \path
            (A) edge [post] node {} (t1)
            (t1) edge [post] node {} (C)
            (A) edge [post,-o] node {} (t2)
            (t2) edge [post] node {} (D)
            (B) edge [post] node {} (t2)
            (B) edge [post] node {} (t3)
            (t3) edge [post] node {} (E)
        ;

        \end{tikzpicture}
}   \caption{A PITPN and its valuation.}
    \label{fig:example}
\end{figure}
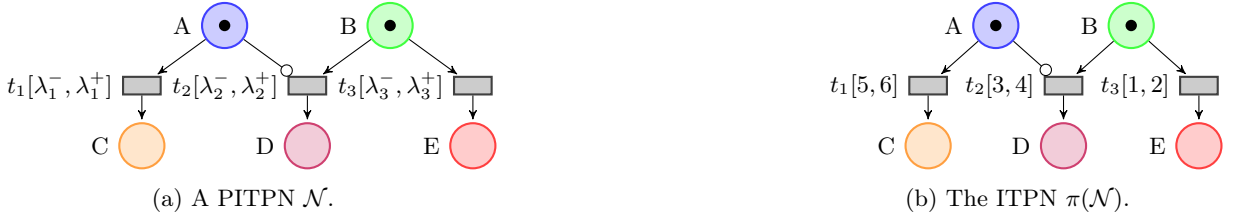

A \emph{marking} of  $\PN$ is an element
$\marking\in\grandn^P$, where $\marking(\place)$
is the number of tokens in place $\place$.
$\valuate{\PN}{\py}$ denotes the  ITPN where each
occurrence of $\lambda_i$ in the PITPN $\PN$  has been replaced by
$\pi(\lambda_i)$ for a parameter valuation $\pi$.  
For example, the ITPN in Fig.~\ref{fig:example:PITPN:valuated} corresponds
to the PITPN in Fig.~\ref{fig:example:PITPN:param} where the parameters
are instantiated with
$\py = \{ \param_1^- \rightarrow 5 , \param_1^+ \rightarrow 6 ,
\param_2^- \rightarrow 3 , \param_2^+ \rightarrow 4 , \param_3^-
\rightarrow 1 , \param_3^+ \rightarrow 2 \}$.

The \emph{concrete semantics} of a PITPN $\PN$ is defined
in terms of concrete ITPNs $\valuate{\PN}{\py}$ where
$\pi\models\Kinit$. 
We say that a transition~$\transition$ is \emph{enabled} in~$\marking$ if $\marking
\geq \relPre{\transition}$ (the number of tokens in
$\marking$ in each input place of $\transition$ is greater than or
equal to the value on the arc between this place and $t$). 
A transition~$\transition$ is \emph{inhibited} if the place connected to
one of its inhibitor arcs is marked with at least as many tokens as
the weight of the inhibitor arc.
A transition~$\transition$ is \emph{active} if it is enabled and not inhibited.
The sets of enabled and inhibited transitions in marking $\marking$ are denoted
$\enabled(\marking)$ and $\inhibited(\marking)$, respectively.
Transition~$\transition$ is \emph{firable} if it has been
(continuously) enabled 
for at least time~$\leftEP{\parIntervalStatic(\transition)}$, without
counting the time it has been inhibited. 
Transition $\transition$ is \emph{newly enabled} by the firing of
transition $\transition_f$ 
in $\marking$
if it is enabled in the resulting marking
$\marking' = \marking - \relPre{\transition_f} +
\relPost{\transition_f}$ but was not enabled
in $\marking - \relPre{\transition_f}$:
\[\mbox{\newlyEnabled}(\transition,\marking,\transition_f) =
(\relPre{\transition} \leq
    \marking  - \relPre{\transition_f} + \relPost{\transition_f})
\land ((\transition=\transition_f)
        \lor \neg (\relPre{\transition} \leq \marking -\relPre{\transition_f})).\]
$\newlyEnabled(\marking,\transition_f)$ denotes the  transitions newly enabled
by  firing  $\transition_f$ in $\marking$.

The semantics of an ITPN is defined as a transition system with states
$(\marking,\interval)$,  where 
$\marking$ is a marking and
$\interval$  is a function mapping each transition enabled in
$\marking$ to
a time interval, and
two kinds of transitions: 
\emph{time} transitions
where time elapses, and discrete transitions when a transition in the net is fired.

\begin{definition}[Semantics of an
  ITPN~\cite{paris-paper}] \label{def:pnet-semantics} 
The dynamic behaviors of  an ITPN $\pi(\PN)$ are defined by the
transition system
$\mathcal{S_{\pi(\PN)}} = (\mathcal{A},a_0,\rightarrow)$, where: 
$\mathcal{A}=\grandn^P\times [\Transition\rightarrow\Interval(\grandq)]$,
 $a_0=(\marking_0,\parIntervalStatic)$ and 
    $(\marking,\interval)
    \rightarrow (\marking',\interval')$ if there exist $\delta\in
    \grandrplus$,  $\transition\in T$, and state
    $(\marking'',\interval'')$ such that
    $(\marking,\interval)\fleche{\delta} (\marking'',\interval'')$ and
    $(\marking'',\interval'') \fleche{\transition}
            (\marking',\interval')$, for the following  relations: 
    \begin{itemize}
        \item the \emph{time transition relation}, defined 
        $\forall\delta\in\grandrplus$ by:\\
        $(\marking,\interval)\fleche{\delta}
            (\marking,\interval')$ iff $\forall \transition \in \Transition$: \\
            $\left\{\begin{array}{l}
                \interval'(t)=\left\{\begin{array}{l}
                    \interval(\transition)\mbox{ if }
                        \transition \in \enabled(\marking) \mbox{ and }
                        \transition\in\inhibited(\marking)\\
                    \leftEP{\interval'(\transition)}=\max(0,\leftEP{\interval(\transition)}
                                       - \delta), 
                    \mbox{ and } \rightEP{\interval'(\transition)} =
                                       \rightEP{\interval'(\transition)}
                                       - \delta 
                    \mbox{ otherwise}
                   \end{array}\right.\\
                \marking \geq\relPre(\transition) \implies
                    \rightEP{\interval'(\transition)}\geq 0
            \end{array}\right.$
        \item the \emph{discrete transition relation}, defined
        $\forall\transition_f\in\Transition$ by:
        $(\marking,\interval)\fleche{\transition_f}
            (\marking',\interval')$ iff\\
            $\left\{\begin{array}{l}
                \transition_f\in\enabled(\marking)\land
                    \transition_f\not\in\inhibited(\marking) \land
                    \marking'=\marking-\relPre{\transition_f}
                        +\relPost{\transition_f} \land
                    \leftEP{\interval(\transition_f)}=0\\
                    \forall\transition\in\Transition,
                        \interval'(\transition)=
                            \left\{\begin{array}{l}
                                \parIntervalStatic(\transition)
                                    \mbox{ if }
                                    \newlyEnabled(\transition,\marking,
                                        \transition_f) \\
                                \interval(\transition) \mbox{ otherwise}
                            \end{array}\right.
            \end{array}\right.$
    \end{itemize}
\end{definition}

The \emph{symbolic} semantics of PITPNs is given
in~\cite{EAGPLP13} as a transition system $(\grandn^P \times \setP,
(\marking_0, K_0), \Fleche{})$  
on \emph{state classes}, i.e.,   pairs $\class = (\marking,\constraint)$
consisting of a marking $\marking$ and  a constraint $\constraint$
over $\Param$.
The firing of a transition leads to
a new marking as in the concrete semantics, and also captures the new constraints
induced by the time that has passed for the transition to fire. For
example, for the PITPN in Fig.~\ref{fig:example:PITPN:param}, 
 the initial class is $(\{A, B\}, \param_1^- \leq \param_1^+ 
\land \param_2^- \leq \param_2^+ \land \param_3^- \leq \param_3^+)$.
When firing transition $\transition_1$, the time spent for
$\transition_1$ to be firable is such that the other transitions
($\transition_3$ in this case) do
not miss their deadlines. So we obtain an additional inequality
$\param_1^-\leq\param_3^+$ and the new state class, obtained after firing
$\transition_1$ is $(\{\textcolor{red}{C}, B\}, \param_1^-\leq\param_1^+ \land
\param_2^-\leq\param_2^+ \land \param_3^-\leq\param_3^+
\land \textcolor{red}{\param_1^-\leq\param_3^+})$. See~\cite{EAGPLP13}
for details.

\subsection{Rewriting with SMT and Maude}
\label{sec:rew-smt}

\paragraph{Rewrite Theories.}

A \emph{rewrite theory} \cite{Mes92} is
a tuple $\mathcal{R} = (\Sigma, E, L, R)$
such that

\begin{itemize}
    \item $\Sigma$ is a signature that declares sorts, subsorts, and function symbols;
    \item $E$ is a set of equations of the form $t=t' \mbox{ \textbf{if} } \psi$,
    where 
    $t$ and $t'$ are terms of the same sort,
    and $\psi$ is a conjunction of equations;

    \item $L$ is a set of \emph{labels};
    and
    
    \item $R$ is a set of rewrite rules
    of the form
    $l : q \longrightarrow r \mbox{ \textbf{if} } \psi$,
    where $l \in L$ is a label,
    $q$ and $r$ are terms of the same sort,
    and
    $\psi$ is a conjunction of equations.

\end{itemize}

$T_{\Sigma, s}$ denotes the set of ground (\ie{} not containing variables)
terms of sort $s$,
 and $T_{\Sigma}(X)_s$ the set of terms of sort $s$
 over a set of variables $X$. $T_{\Sigma}(X)$ and
 $T_{\Sigma}$ denote all terms and ground terms, respectively.
 A substitution $\sigma : X \rightarrow T_{\Sigma}(X)$
 maps each variable to a term of the same sort,
 and
  $t \sigma$ 
  denotes 
 the term obtained
by simultaneously replacing each variable $x$ in a term $t$ with $\sigma(x)$.
The domain of a substitution 
$\sigma$ is $\mathit{dom}(\sigma) = \{x \in X \mid \sigma(x) \neq x\}$,
assumed to be finite. 

A \emph{one-step rewrite} $t \longrightarrow_{\mathcal{R}} t'$ holds 
if there are
a rule $l : q \longrightarrow r \mbox{ \textbf{if} } \psi$,
a subterm $u$ of $t$,
and a substitution $\sigma$ such that
$u = q\sigma$ (modulo equations),
$t'$ is the term obtained from $t$
by replacing $u$ with $r\sigma$,
and $v\sigma = v'\sigma$ holds
for each $v = v'$ in $\psi$.
We denote by
$\longrightarrow_{\mathcal{R}}^\ast$
the reflexive-transitive closure of $\longrightarrow_{\mathcal{R}}$.

A rewrite theory $\mathcal{R}$
is called \emph{topmost}
iff
there is a sort $\mathit{State}$
at the top of one of the connected components of the subsort partial order
such that
for each rule $l : q \longrightarrow r \mbox{ \textbf{if} } \psi$,
both $q$ and $r$ have the top sort $\mathit{State}$,
and
no operator has sort $\mathit{State}$
or any of its subsorts as an argument sort.

\paragraph{Rewriting with SMT \cite{rocha-rewsmtjlamp-2017}.}
For a signature $\Sigma$ and 
a set of equations $E$,
a \emph{built-in theory} $\mathcal{E}_0$
is a first-order theory with a signature $\Sigma_0 \subseteq \Sigma$,
where
(1) each sort $s$ in $\Sigma_0$ is minimal in $\Sigma$;
(2) $s \notin \Sigma_0$ for each operator $f:s_1\times \cdots \times s_n \rightarrow s$
    in $\Sigma \setminus \Sigma_0$; and
(3) $f$ has no other subsort-overloaded typing in $\Sigma_0$.
The satisfiability of a constraint in $\mathcal{E}_0$
is assumed to be decidable
using the SMT theory $\mathcal{T}_{\mathcal{E}_0}$
which is
consistent with  $(\Sigma, E)$, i.e., 
for $\Sigma_0$-terms $t_1$ and $t_2$,
if $t_1= t_2$ modulo $E$, then $\mathcal{T}_{\mathcal{E}_0} \models t_1 = t_2$.

A \emph{constrained term}
is a pair $\phi \parallel t$ of
a constraint $\phi$ in $\mathcal{E}_0$ and  a term $t$ in $T_{\Sigma}(X_0)$
over variables $X_0 \subseteq X$ of the built-in sorts in
$\mathcal{E}_0$ \cite{rocha-rewsmtjlamp-2017,bae2019symbolic}.
A constrained term
 $\phi \parallel t$
\emph{symbolically} represents
all instances of the pattern
$t$
such that $\phi$ holds:
    $\llbracket \phi \parallel t \rrbracket
=
\{t' \mid t' = t\sigma \ \mbox{(modulo $E$) and}\  \mathcal{T}_{\mathcal{E}_0}
\models \phi\sigma \ \mbox{for ground}\ \sigma : X_0 \to T_{\Sigma_0}
\}.
$

An \emph{abstraction of built-ins}
for a $\Sigma$-term $t \in T_{\Sigma}(X)$
is a pair $(t^\circ, \sigma^\circ)$
of 
a term $t^\circ \in T_{\Sigma \setminus \Sigma_0}(X)$
and
a substitution $\sigma^\circ : X_0 \to T_{\Sigma_0}(X_0)$
such that
$t = t^\circ \sigma^\circ$
and  $t^\circ$  contains no duplicate variables in $X_0$.
Any non-variable built-in subterms of $t$ are 
replaced by distinct built-in variables in $t^\circ$.
$\Psi_{\sigma^\circ} = \bigwedge_{x \in \mathit{dom}(\sigma^\circ)} x = x \sigma^\circ $.
Let $\phi \parallel t$ be a constrained term 
and $(t^\circ, \sigma^\circ)$ an
abstraction of built-ins for $t$.
If $\mathit{dom}(\sigma^\circ) \cap \ovars{\phi \parallel t} = \emptyset$,
then
$\llbracket \phi \parallel t \rrbracket
= 
\llbracket \phi \wedge \Psi_{\sigma^\circ} \parallel t^\circ \rrbracket$ 
\cite{rocha-rewsmtjlamp-2017}

Let $\mathcal{R}$ be a topmost theory
such that for each rule $l : q \longrightarrow r \mbox{ \textbf{if} } \psi$, extra variables not occurring in the left-hand side $q$
are in $X_0$,
and
$\psi$ is a constraint in a built-in theory $\mathcal{E}_0$.
A \emph{one-step symbolic rewrite} $\phi \parallel t
\rightsquigarrow_{\mathcal{R}} \phi' \parallel t'$ holds
iff there exist
a rule $l : q \longrightarrow r \mbox{ \textbf{if} } \psi$
and
a substitution $\sigma : X \to T_{\Sigma}(X_0)$
such that
(1) $t = q\sigma$
and
$t' = r\sigma$
(modulo equations),
(2) $\mathcal{T}_{\mathcal{E}_0} \models (\phi \wedge  \psi
      \sigma) \Leftrightarrow \phi'$, and
(3) $\phi'$ is $\mathcal{T}_{\mathcal{E}_0}$-satisfiable.
We denote by
$\rightsquigarrow_{\mathcal{R}}^\ast$
the reflexive-transitive closure of $\rightsquigarrow_{\mathcal{R}}$.

A \emph{symbolic rewrite}
on constrained terms
symbolically represents
a (possibly infinite) set of system transitions.
If $\phi_t \parallel t \rightsquigarrow^\ast \phi_u \parallel u$
is a symbolic rewrite,
then there exists a ``concrete'' rewrite  $t' \longrightarrow^\ast u'$
 with $t' \in \llbracket \phi_t \parallel t \rrbracket$
 and $u' \in \llbracket \phi_u \parallel u \rrbracket$.
Conversely,
for any concrete
rewrite $t' \longrightarrow^\ast u'$ with
$t' \in \llbracket \phi_t \parallel t \rrbracket$,
there exists
a symbolic rewrite $\phi_t \parallel t \rightsquigarrow^\ast \phi_u \parallel u$
with $u' \in \llbracket \phi_u \parallel u \rrbracket$.

\paragraph{Maude.}
Maude~\cite{maude-book} is a language and tool
supporting the specification and analysis of  rewrite theories.
We summarize its syntax below:

\begin{maude}
pr R .                   --- Importing a theory R
sorts S ... Sk .         --- Declaration of sorts S1,..., Sk
subsort S1 < S2 .        --- Subsort relation
vars X1 ... Xm : S .     --- Logical variables of sort S
op f : S1 ... Sn -> S .  --- Operator S1 x ... x Sn -> S
op c : -> T .            --- Constant c of sort T
eq t = t' .              --- Equation
ceq t = t' if c .        --- Conditional equation
crl [l] : q => r if c .  --- Conditional rewrite rule
\end{maude}

\noindent Maude provides a number of analysis methods,  including computing the
normal form of a term $t$ (command \lstinline[mathescape]{red $t$}),
simulation by rewriting (\lstinline[mathescape]{rew $t$}) and rewriting following a given 
strategy (\lstinline[mathescape]{srew $t$ using $str$}). 
Basic strategies include
$r\mathtt{[}\sigma\mathtt{]}$ (apply  rule $r$ once with the optional ground
substitution $\sigma$), \code{all} (apply any of the rules once), 
and \code{match $P$ s.t. $C$} that checks
whether the current term matches the pattern $P$ subject to the constraint $C$.
Compound strategies can be defined using concatenation
($\alpha\,;\,\beta$), disjunction ($\alpha\, |\, \beta$), iteration ($\alpha \,\mathtt{*}$),
$\alpha \code{ or-else } \beta$ (execute $\beta$ if $\alpha$ fails), 
normalization  $\alpha\,\mathtt{!}$ (execute $\alpha$  until it cannot be further
applied), etc. 

Maude also offers explicit-state
reachability analysis from a ground
term $t$ (\lstinline[mathescape]{search [$n$,$m$] $t$ =>* $t'$ such that $\Phi$})
and  model checking an LTL formula $F$
(\lstinline[mathescape]{red modelCheck($t$, $F$)}). Atomic propositions
in $F$ are user-defined terms of sort \code{Prop}, and the function
\lstinline{op _|=_ : State Prop -> Bool} specifies which states satisfy a given 
proposition. LTL formulas are then built from state formulas, boolean connectives
and the temporal logic operators \texttt{[]} (``always''), \texttt{<>} (``eventually'')
and  \texttt{U} (``until''). 
For symbolic reachability analysis, the command 

\begin{maude}
smt-search [$n$, $m$]: $t$ =>* $t'$ such that $\Phi$  --- n and m are optional
\end{maude}
symbolically searches for $n$ states,
reachable from $t \in T_{\Sigma}(X_0)$ within $m$ steps,
that match the pattern $t' \in T_{\Sigma}(X)$ and satisfy
the constraint $\Phi$ in $\mathcal{E}_0$.
More precisely,
it searches for
a constrained term $\phi_u \parallel u$
such that
$\mathit{true} \parallel t \rightsquigarrow^\ast \phi_u \parallel u$
and  for some $\sigma : X \to T_{\Sigma}(X)$,
$u = t'\sigma$ (modulo equations) and
$\mathcal{T}_{\mathcal{E}_0} \models \phi_u \Rightarrow \Phi\sigma$.

Maude
provides
  built-in  sorts
\code{Boolean}, \code{Integer}, and \code{Real} for
the SMT theories of Booleans, integers, and reals. Rational constants of sort \code{Real} are written
\code{$n$/$m$} (e.g., \code{0/1}).
Maude-SE~\cite{maude-se} extends Maude with
additional functionality
for rewriting modulo SMT,
including
witness generation for \lstinline{smt-search}.
It uses two theory transformations
to implement
symbolic rewriting~\cite{rocha-rewsmtjlamp-2017}.
In essence,
a  rewrite
rule $l : q \longrightarrow r \mbox{ \textbf{if} } \psi$
is transformed
into
a constrained-term rule
\begin{align*}
l:
\mathtt{PHI} \parallel q^\circ
\longrightarrow
(\mathtt{PHI} \mathbin{and} \psi \mathbin{and} \Psi_{\sigma^\circ} ) \parallel r
\mbox{ \textbf{if} }
&
\mathtt{smtCheck}(\mathtt{PHI} \mathbin{and} \psi \mathbin{and} \Psi_{\sigma^\circ} )
\end{align*}

\noindent
where $\mathtt{PHI}$ is a \code{Boolean} variable,
$(q^\circ, \sigma^\circ)$ is an abstraction of built-ins for $q$,
and
\code{smtCheck} invokes the underlying  SMT solver
to check the satisfiability of an SMT condition.
This rule is executable
if the extra SMT variables in $(\ovars{r} \cup
\ovars{\psi} \cup \ovars{\Psi_{\sigma^\circ}}) \setminus \ovars{q^\circ}$
are considered constants.

\section{A Rewriting Logic Semantics for ITPNs }\label{sec:tranformation}
\label{sec:concrete}

This section presents a rewriting logic semantics for (non-parametric)
ITPNs, using 
 a (non-executable) rewrite theory
$\rtheorySem$. We provide a 
bisimulation  relating the concrete 
semantics of a net $\PN$ and a rewrite relation in $\rtheorySem$, and  discuss  variants of $\rtheorySem$
to avoid  consecutive tick steps and  to enable time-bounded
 analysis.

\subsection{Formalizing  ITPNs in Maude: The Theory $\rtheorySem$}
\label{subsec:eq-theory}

We fix $\PN$ to be the ITPN $\tuple{\Place, \Transition, \emptyset, \relPre{(.)},
\relPost{(.)}, \relInhib{(.)}, \markingInit, \parIntervalStatic,
true}$, and show how ITPNs and markings of such nets can be
represented as Maude terms.

We first  define sorts for representing transition
 labels, places, and time values  in Maude.
 The usual approach is to represent each transition $t_i$ and
 each place $p_j$ as a constant of sort \texttt{Label} and
 \texttt{Place}, respectively (e.g., \texttt{ops \(p_1\) \(p_2\)
   ... \(p_m\) : -> Place [ctor]}).  To avoid even this simple
 parameterization and just use a single rewrite theory $\mathcal{R}_0$
 to define the semantics of  all ITPNs, we assume
 that  places and transition (labels) can be represented as
 strings. Formally, we assume that there is an injective naming
 function $\eta : P \cup T \rightarrow \mathtt{String}$; to avoid
 cluttering the paper with subscripts, we usually do not mention
 $\eta$  explicitly.

\begin{maude}
  protecting STRING .   protecting RAT .
  sorts Label Place .   --- identifiers for transitions and places  
  subsorts String < Label Place .   --- we use strings for simplicity 
  sorts Time TimeInf .  --- time values 
  subsort Zero PosRat < Time  < TimeInf .
  op inf : -> TimeInf [ctor] .
  vars T T1 T2 : Time     .
  eq T <= inf = true .
\end{maude}

The sort \texttt{TimeInf}
 adds an ``infinity'' value \texttt{inf} to the sort \texttt{Time} of
 time values, which are the  non-negative rational numbers (\code{PosRat}). 

The ``standard'' way of formalizing Petri nets in rewriting logic (see,
e.g.,~\cite{Mes92,petri-nets-in-maude}) represents, e.g.,  a marking
with two tokens in place $p$ 
and three tokens in place $q$  as the Maude term
$p\; p\; q\; q\; q$. This is crucial to support \emph{concurrent}
firings of transitions in a net. However, since the semantics of PITPNs is
an \emph{interleaving} semantics, and   to support
rewriting-with-SMT-based analysis from
parametric initial markings (\Cref{ex:pmarking}),
we instead represent markings as 
maps from places to the number of tokens in that place,
so that the above 
marking is  represented by the Maude
term  \texttt{\(\nameit(p)\)\,|->\,2 ; \(\nameit(q)\)\,|->\,3}. 

The following declarations define  the 
sort  \code{Marking} to consist of  \code{;}-separated sets of pairs \code{\(\nameit(p)\)\,\,|->\,\,\(n\)}. 
Time intervals are represented as  terms \code{[\(\mathit{lower}\,\):\(\,\mathit{upper}\)]} where the upper
bound $\mathit{upper}$,   
of sort \code{TimeInf}, also can be the infinity value \code{inf}. 
The  Maude term
\code{  \(\;\nameit(t)\)\,: \(pre\) --> \(post\) inhibit \(inhibit\) in
  \(interval\) } represents a transition $t\in T$,  
where \code{\(pre\)}, \code{\(post\)},  and  \code{\(inhibit\)} are markings
representing, respectively, 
  $\relPre{(t)}, \relPost{(t)}, \relInhib{(t)}$; and
  \code{\(interval\)} represents the interval $J(t)$.  A \code{Net}
  is represented as a \texttt{;}-separated set of such transitions
  (lines 11--12):

\begin{maude}
sort Marking . --- Markings
op empty : -> Marking [ctor] .
op _|->_ : Place Nat -> Marking [ctor] .
op _;_ : Marking Marking -> Marking [ctor assoc comm id: empty] .
sort Interval .  --- Time intervals (the upper bound can be infinite)  
op `[_:_`] : Time TimeInf -> Interval [ctor] .
sorts Net Transition .    --- Transitions and nets
subsort Transition < Net .
op _`:_-->_inhibit_in_ :
  Label Marking Marking Marking Interval -> Transition [ctor] .
op emptyNet : -> Net [ctor] . 
op _;_ : Net Net -> Net [ctor assoc comm id: emptyNet] .
\end{maude}

  \begin{example}
Assuming the obvious naming function $\nameit$ mapping $A$ to
\texttt{"A"}, and so on, the net in \Cref{fig:example} is
represented as the following term of sort \code{Net}:

\begin{maude}
"t1" : ("A" |-> 1) --> ("C" |-> 1) in [5 : 6] ; 
$\pTransitionM{"t2"}{("B" |-> 1)}{("D" |-> 1)}{("A" |-> 1)}{3 : 4}$ ;
"t3" : ("B" |-> 1) --> ("E" |-> 1) in [1 : 2]$.$
\end{maude}
\end{example}

We  define some useful operations on markings, such as
\texttt{_+_} and \texttt{_-_}:

\begin{maude}
vars N1 N2 : Nat .    vars M M1 M2 : Marking .   var P : Place .
op _+_ : Marking Marking -> Marking .
eq ((P |-> N1) ; M1) + ((P |-> N2) ; M2) = (P |-> N1 + N2) ; (M1 + M2) .
eq M1 + empty = M1 .
\end{maude}

\noindent (This definition  assumes that 
  each place  
in \code{\(\marking2\)} appears  once in \code{\(\marking1\)} and
\code{\(\marking1\:\)+\(\:\marking2\)}.) The 
function \code{_-_} on markings is defined similarly.
The following functions  compare markings and  check whether a transition is
active in a  marking:

\begin{maude}
op _<=_ : Marking Marking -> Bool .   --- Comparing markings
eq ((P |-> N1) ; M1) <= ((P |-> N2) ; M2)  =   N1 <= N2 and (M1 <= M2) .
eq empty <= M2 = true .
ceq M1 <= empty = false if M1 =/= empty .
op active : Marking Transition -> Bool .  --- Active transition
eq active(M, L : PRE --> POST inhibit INHIBIT in INTERVAL) =
      (PRE <= M) and not inhibited(M, INHIBIT) .
op inhibited : Marking Marking -> Bool .  --- Inhibited transition
eq inhibited(M, empty) = false .
eq inhibited((P |-> N2) ; M, (P |-> N) ; INHIBIT) =
      ((N > 0) and (N2 >= N)) or inhibited(M, INHIBIT) .
\end{maude}

\label{sec:states}

\paragraph{Dynamics.}  We  define the dynamics of ITPNs as a Maude
``interpreter'' for such nets. 
The concrete ITPN semantics in~\cite{paris-paper} dynamically adjusts
the ``time intervals'' of non-inhibited transitions when time
elapses.
Unfortunately, the definitions in~\cite{paris-paper} seem slightly
contradictory: On the one hand, time interval end-points should be
non-negative, and only enabled transitions have intervals in the
states; on the other hand, the definition of time and discrete
transitions in~\cite{paris-paper} mentions $\forall t\in T, I'(t) =
...$ and $\marking \geq\relPre(\transition) \implies
                    \rightEP{\interval'(\transition)}\geq 0$, which
                    seems superfluous if all end-points are
                    non-negative. Taking the definition of time and
                    transition steps in~\cite{paris-paper} (our
                    Definition~\ref{def:pnet-semantics}) leads us to
                    time intervals where 
the right end-points of disabled transitions
could have \emph{negative}
values.  This has some
disadvantages: (i) ``time values'' can be negative numbers; (ii) we
have counterintuitive  ``intervals'' $[0, -r]$ where the right end-point is smaller
than the left end-point; (iii) the reachable ``state spaces'' (in suitable
discretizations) could  be infinite when these negative
values could be unbounded.

To avoid these ``inconsistencies'',  and to have a simple and
well-defined semantics,   we  use 
``clocks'' instead of ``decreasing intervals'';  a clock denotes  how
long the corresponding 
transition has been enabled (but not inhibited). Furthermore, to
reduce the state space, the clocks of disabled transitions  are always
zero.   The resulting semantics is equivalent to the (most natural interpretation of the) one in~\cite{paris-paper} in a way made
precise in Theorem~\ref{thm:ground}.   

The sort \code{ClockValues} denotes sets of \texttt{;}-separated  terms
\code{\(\nameit(t)\)\,->\,\(\tau\)}, where \(t\) is the (label of the) transition and
\code{\(\tau\)}  represents
the current value of  $t$'s ``clock.'' 

\begin{maude}
sort ClockValues . --- Values for clocks
op empty : -> ClockValues [ctor] .
op _->_ : Label Time -> ClockValues [ctor] .
op _;_ : ClockValues ClockValues -> ClockValues [ctor assoc comm id:$\:$empty] .
\end{maude}  

The states in $\rtheorySem$ are
terms $m$\,\texttt{:}\,$\mathit{clocks}$\,\texttt{:}\,$\mathit{net}$
of sort \code{State}, where $m$ represents the current marking,
$\mathit{clocks}$ the current values of the transition clocks, and
$\mathit{net}$ the representation of the Petri net:

\begin{maude}
sort State .
op _:_:_ : Marking ClockValues Net -> State [ctor] .
\end{maude}

 The following rewrite rule models the application of a transition
\code{L} in the net
\code{(\pTransitionMI{L}{PRE}{POST}{INHIBIT}{INTERVAL}) ; NET'}.
 Since \texttt{_;_}
is declared to be associative and commutative, \emph{any} transition \code{L}
in the net
can be applied using this rewrite rule: 

\begin{maude}
crl [applyTransition] :
     M  :  (L -> T) ; CLOCKS  :
     ($\pTransitionMI{L}{PRE}{POST}{INHIBIT}{INTERVAL}$) ; NET
 =>  $\highlight{(M - PRE) + POST}$ :
     $\highlight{L -> 0}$ ; updateClocks(CLOCKS, M$\:$-$\:$PRE, NET) :
     ($\pTransitionMI{L}{PRE}{POST}{INHIBIT}{INTERVAL}$) ; NET'
 if $\highlight{active}$(M, $\pTransitionMI{L}{PRE}{POST}{INHIBIT}{INTERVAL}$)
    and ($\highlight{T in INTERVAL}$) .

op _in_ : Time Interval -> Bool .
eq T in [T1 : T2] = (T1 <= T) and (T <= T2) .
eq T in [T1 : inf] = T1 <= T .
\end{maude}

\noindent The  transition \code{L} is active (enabled and not inhibited)
in the  marking \code{M} 
 and its clock value \code{T} is in the \code{INTERVAL}.
After performing the transition, the  marking is
\code{(M\,-\,PRE)\,+\,POST},
the clock of \code{L} is reset\footnote{Since in our semantics clocks
  of disabled transitions should be zero, we can safely set \texttt{L}
  to \texttt{0} in this rule.} 
and the other clocks 
    are updated using the following function:

\begin{maude}
eq updateClocks((L' -> T') ; CLOCKS, INTERM-M, 
                ($\pTransitionMI{L'}{PRE}{POST}{INHIBIT}{INTERVAL}$) ; NET) 
 = if PRE <= INTERM-M then (L' -> T') else (L' -> 0) fi ;
   updateClocks(CLOCKS, INTERM-M, NET) .
eq updateClocks(empty,  INTERM-M, NET) = empty .   
 \end{maude}

The second rewrite rule in $\rtheorySem$ specifies how time
advances. Time can advance by \emph{any} value \code{T}, as long as
time 
does not advance beyond the time when an active transition must be
taken.  The clocks are updated according to the elapsed time \code{T},
except for those transitions that are disabled or inhibited: 

\begin{maude}
crl [tick] : M$\;\,$:$\;\,$CLOCKS$\;\,$:$\;\,$NET  =>  M$\;\,$:$\;\,$increaseClocks(M,$\,$CLOCKS,$\,$NET,$\,\highlight{T}$)$\;\,$:$\;\,$NET
    if  $\highlight{T}$ <= mte(M, CLOCKS, NET) [nonexec] .
\end{maude}

\noindent This rule is not executable (\code{[nonexec]}), since the 
variable \code{T}, which denotes how much time
advances, only  occurs  
in the right-hand side of the rule. \code{T} is
therefore \emph{not} assigned any value by the substitution matching the rule
with the state being rewritten. 
This time advance  \code{T}
must be less or equal to  the minimum of the upper bounds of the enabled 
transitions in the marking \code{M}:

\begin{maude}
op mte : Marking ClockValues Net -> TimeInf .
eq mte(M, (L$\;$->$\;$T)$\;\,$;$\;\,$CLOCKS, (L$\;\,$:$\;\,$PRE$\;\,$-->$\;\,$POST$\;\,...\;\,$in$\;\,$[T1$\,$:$\,\highlight{inf}$])$\;\,$;$\;\,$NET)
 = mte(M, CLOCKS, NET) .
eq mte(M, (L$\;$->$\;$T)$\;\,$;$\;\,$CLOCKS, (L$\;\,$:$\;\,$PRE --> ... in$\;\,$[T1$\,$:$\,\highlight{T2}$])$\;\,$;$\;\,$NET)
 = if active(M, L : ...) then min(T2 - T, mte(M, CLOCKS, NET))
   else mte(M, CLOCKS, NET) fi .
eq mte(M, empty, NET) = inf .
\end{maude}

The function \code{increaseClocks} increases the transitions clocks according to
the elapsed time, except for those transitions that are disabled or inhibited: 

\begin{maude}
op increaseClocks : Marking ClockValues Net Time -> ClockValues .
eq increaseClocks(M, $\highlight{(L -> T1)}$ ; CLOCKS, (L : PRE --> ...) ; NET, $\highlight{T}$)
 = if active(M, L : PRE --> ...)
   then $\highlight{(L -> T1 + T)}$ else $\highlight{(L -> T1)}$ fi ; increaseClocks(M,$\,$CLOCKS,$\,$NET,$\,$T)$\,$.
eq increaseClocks(M, empty, NET, T) = empty .
\end{maude}

The following function $\encBase{\_}$ formalizes how markings and nets are
represented as terms, of respective sorts \texttt{Marking} and
\texttt{Net}, in rewriting logic.\footnote{$\encBase{\_}$ is parametrized
  by the naming function $\nameit$; however, we do not show this
  parameter explicitly.}

\begin{definition}Let $\PN=\tuple{\Place, \Transition, \emptyset, \relPre{(.)},
\relPost{(.)}, \relInhib{(.)}, \markingInit, \parIntervalStatic,
\mathit{true}}$ be an ITPN.   Then $\encBase{\_} : \mathbb{N}^P \rightarrow
\mathcal{T}_{\mathcal{R}_0,\mathtt{Marking}}$ is defined by
  $\encBase{\{p_1\mapsto n_1, \ldots, p_m\mapsto n_m\}} = \nameit(p_1)
  \texttt{\,|->\;} n_1 \;\texttt{;} \,\ldots\, \texttt{;}\; \nameit(p_m)\;
  \texttt{|->} \;n_m$, where  we can omit  entries $\nameit(p_j)
  \texttt{\,|->\;} 0$. The Maude representation $\encBase{\PN}$ of the net $\PN$ is
  the term $\encBase{t_1}\, \texttt{;}\, \cdots\,
\texttt{;}\, \encBase{t_n} $  
of sort \code{Net}, where, for each $t_i \in T$, $\encBase{t_i}$ is
\newline \code{
$\nameit(t_i)\;$:$\;\encBase{\relPre{(t_i)}}$ -{}->
    $\encBase{\relPost{(t_i)}}$ inhibit
    $\encBase{\relInhib{(t_i)}}$ in [$\leftEP{J(t_i)}$ : $\rightEP{J(t_i)}$]. 
}

\end{definition}

\subsection{Correctness of the Semantics} \label{sec:bisimulation}

In this section we show that our rewriting logic semantics
$\mathcal{R}_0$ correctly simulates any ITPN $\PN$. More concretely, we
provide a bisimulation result relating behaviors from $a_0 = 
(\marking_0,\parIntervalStatic)$
in $\PN$ with behaviors in 
$\mathcal{R}_0$ starting from the initial state
$\encBase{\marking_0}$\,\code{:}\,\code{initClocks($\encBase{\PN}$)}\,\code{:}\,$\encBase{\PN}$, 
where  \texttt{initClocks(\(\mathit{net}\))} is the clock valuation that
assigns the value \texttt{0} to each transition (clock) $\eta(t)$ for
each transition (label) $\eta(t)$ in $\mathit{net}$.

Since a transition in $\PN$ consists of a delay followed by
a discrete transition, we  define a corresponding rewrite
relation $\mapsto$ combining the \code{tick} and
\code{applyTransition} rules, and prove the bisimulation for this
relation.

\begin{definition}\label{def:two-steps}
    Let $t_1, t_2, t_3$ be terms of sort \code{State} in
    $\mathcal{R}_0$. We write $t_1 
    \mapsto t_3$ if  
    there exists a $t_2$ such that $t_1 \longrightarrow t_2$ is a one-step 
    rewrite applying the \code{tick} rule in $\mathcal{R}_0$
    and $t_2 \longrightarrow t_3$  
    is a one-step rewrite applying the \code{applyTransition} rule in
    $\mathcal{R}_0$.  
    Furthermore, we write $t_1 \mapsto^* t_2$ to
    indicate that there exists a 
    sequence of $\mapsto$ rewrites from $t_1$ to $t_2$.
\end{definition}

The following relation relates our clock-based states with the
changing-interval-based states; the correspondence is a
straightforward function, except for the case when the upper bound of
a transition is $\infty$:

\begin{definition}  \label{def:bisim-relation}
  Let $\PN = \tuple{\Place, \Transition, \emptyset, \relPre{(.)},
    \relPost{(.)}, \relInhib{(.)}, \markingInit, \parIntervalStatic,
    \mathit{true}}$
  be an ITPN and $\mathcal{S_{\PN}} = (\mathcal{A},a_0,\rightarrow)$ be its
  concrete semantics.   
   Let
  $T_{\Sigma,\texttt{State}}$ denote the set of $E$-equivalence  
  classes of ground terms of sort \texttt{State} in $\mathcal{R}_0$. 
  We define a relation  $\approx \, \subseteq \mathcal{A} \times T_{\Sigma,
    \texttt{State}}$, relating states in the concrete semantics  
  of $\PN$ to states (of sort \texttt{State}) in
  $\mathcal{R}_0$, where for all states $(\marking,\interval) \in \mathcal{A}$,  
  $(\marking,\interval) \approx m\; \code{:}
  \;\mathit{clocks}\; \code{:} \; \mathit{net}$ 
 if and only if $m = \encBase{M}$ and $\mathit{net} = \encBase{\PN}$
 and  for
  each transition $t\in \Transition$,
  \begin{itemize}
    \item  the value of $\eta(t)$ in $\mathit{clocks}$ is \code{0} if
      $t$ in not enabled in $M$; 
    \item otherwise:
      \begin{itemize}
      \item if $\rightEP{\parIntervalStatic(t)} \not = \infty$ then
the value of clock $\eta(t)$ in $\mathit{clocks}$ is  $\rightEP{\parIntervalStatic(t)}
          - \rightEP{\interval(t)}$;
          \item otherwise, if $\leftEP{\interval(t)} > 0$ then
            $\eta(t)$ has 
            the value $\leftEP{\parIntervalStatic(t)}
          - \leftEP{\interval(t)}$ in $\mathit{clocks}$; otherwise, the  value
          of $\eta(t)$ in 
          $\mathit{clocks}$ could be any value $\tau\geq
          \leftEP{\parIntervalStatic(t)}$.
        \end{itemize}
        \end{itemize}
\end{definition}

\begin{theorem}
\label{thm:ground}
Let $\PN = \tuple{\Place, \Transition, \emptyset, \relPre{(.)},
  \relPost{(.)}, \relInhib{(.)}, \markingInit, \parIntervalStatic,
  true}$ be an ITPN, 
 and $\mathcal{R}_0 = (\Sigma,E,L,R)$. Then, 
$\approx$ is a bisimulation between the transition systems
$\mathcal{S_{\PN}} = (\mathcal{A},a_{0},\rightarrow)$ and $\left(T_{\Sigma
    , \texttt{State}}, (\encBase{\marking_0}
  \;\code{:}\;\mathtt{initClocks}(\encBase{\PN})\;\code{:}\;\encBase{\PN}),   
\mapsto\right)$.
\end{theorem}

\subsection{Some Variations of $\rtheorySem$}\label{sec:optim} 

This section introduces the theories $\rtheorySemN{1}$ and $\rtheorySemN{2}$, two variations of $\rtheorySem$ to
 reduce the reachable state space (in symbolic analyses) and  to enable time-bounded
analysis.
   $\rtheorySemN{1}$ avoids  consecutive application of
  the \texttt{tick} rule. This  is  useful for \emph{symbolic}
    analysis since in concrete executions of $\rtheorySemN{1}$, a tick rule
    application may not advance time far enough for a transition to become
    enabled, leading to a deadlock.
$\rtheorySemN{2}$
      adds a ``global clock'', denoting how much time has elapsed in the system. (In
    $\rtheorySem$ such a global clock can also be encoded by the
    clock of a
    ``new'' transition which is never enabled).  This allows for 
  analyzing  time-bounded properties (can a certain state be reached
  in a certain time interval?).

\paragraph{The Theory $\rtheorySemN{1}$.} To avoid consecutive tick
rule applications, we can add a new component---whose value is
either \texttt{tickOk} or \texttt{tickNotOk}---to the global
state. The tick rule can only be applied when this new component of the
global state has the value \texttt{tickOk}. We therefore add a new
constructor \verb@_:_:_:_@ for these extended global states, a new
sort \texttt{TickState} with values \texttt{tickOk} and
\texttt{tickNotOk}, and modify (or add) the two rewrite rules below:

\begin{maude}
sort TickState .
ops tickOk tickNotOk : -> TickState [ctor] .
op _:_:_:_ : TickState Marking ClockValues Net -> State [ctor] .
var TS : TickState .
crl [applyTransition] :
 $\highlight{TS}$ : M : ((L -> T) ; CLOCKS) : (L : PRE --> ...) ; NET) =>
 $\highlight{tickOk}$ : ((M - PRE) + POST) :  ... if active(...) and (T in INTERVAL) .
crl [tick] : $\highlight{tickOk}$ : M : ...  => $\highlight{tickNotOk}$ : M : increaseClocks(...) ...
 if  T <= mte(M, CLOCKS, NET) [nonexec] .
\end{maude}

\begin{theorem}\label{th:r0r1}
    Let  $t= m\; \code{:} \;\mathit{clocks}\; \code{:} \; \mathit{net}$
    be a term of sort \code{State} in 
  $\rtheorySem$. Then, 

  $t  \longrightarrow^*_{\rtheorySem} m'\; \code{:} \;\mathit{clocks'}\; \code{:} \; \mathit{net}$
iff 

    \noindent$\mathtt{tickOk}\; \code{:}\; m\; \code{:} \;\mathit{clocks}\; \code{:} \; \mathit{net}
    \longrightarrow^*_{\rtheorySemN{1}}
    \mathtt{tickNotOk}\; \code{:}\; m'\; \code{:} \;\mathit{clocks'}\; \code{:} \; \mathit{net}$.
\end{theorem}
Although reachability is preserved, an ``arbitrary'' application of
the tick rule in  $\rtheorySemN{1}$, where time does not advance far
enough for a transition to be taken, could lead to a deadlock in
$\rtheorySemN{1}$ but not in
$\rtheorySemN{0}$.

\paragraph{The Theory $\rtheorySemN{2}$.}  To answer questions such as
whether a certain state can be reached in a certain time interval, and
to enable time-bounded analysis where behaviors  beyond the time
bound are not explored, we
add a new component, denoting the ``global time,''  to the global state:

\begin{maude}
op _:_:_:_$\highlight{@_}$ :  TickState Marking ClockValues Net $\highlight{Time}$ -> State [ctor] .
\end{maude}

The \code{tick} and \code{applyTransition} rules are modified as expected. For instance, 
the rule \code{tick} becomes: 
\begin{maude}
crl [tick] :   tickOk    : M : CLOCKS : NET $\highlight{@ GT}$
          =>   tickNotOk : M : increaseClocks(..., T) : NET $\highlight{@ GT + T}$
if  T <= mte(M, CLOCKS, NET) [nonexec] .
\end{maude}

\noindent where \code{GT} is  a variable of sort \code{Time}.  For a
time bound $\Delta$, we can add a conjunct \code{GT\,+\,T\,<=}$\;\Delta$
in the condition of this rule
to stop executing beyond the time bound. 

Let $t$ and $t'$ be terms of sort \code{State} in $\rtheorySem$. We say that
$t'$ is reached in time $d$  from $t$, written $t
\dtransition^*_{\rtheorySem} t'$, if $t \longrightarrow^*_{\rtheorySem} t'$ and
$d$ is the sum of the values taken by the variable \code{T} in the different
applications of the rule \code{tick} in such a trace.

 \begin{theorem}\label{th:r0r2}
    Let  $t= m\; \code{:} \;\mathit{clocks}\; \code{:} \; \mathit{net}$
    be a term of sort \code{State} in 
  $\rtheorySem$. Then, 
    $t \dtransition^*_{\rtheorySem} m'\; \code{:} \;\mathit{clocks'}\; \code{:} \; \mathit{net}$
    iff 

    \noindent$\mathtt{tickOk}\; \code{:}\; m\; \code{:} \;\mathit{clocks}\; \code{:} \; \mathit{net}\; \code{@}\; 0
    \longrightarrow^*_{\rtheorySemN{2}}
    \mathtt{tickNotOk}\; \code{:}\; m'\; \code{:} \;\mathit{clocks'}\; \code{:} \; \mathit{net}\; \code{@}\; d$.

\end{theorem}

\section{Explicit-state Analysis of ITPNs in Maude}
\label{sec:concrete-ex}

The theories $\mathcal{R}_0$--$\mathcal{R}_2$ cannot be directly
executed in Maude, since the \texttt{tick} rule introduces a new
variable \texttt{T} in its right-hand side.
Following the Real-Time Maude~\cite{tacas08,rtm-journ}
methodology for analyzing dense-time systems, although we cannot cover
all time points, we can choose to ``sample'' system execution at \emph{some}
time points. For example, in this section we change the \texttt{tick}
rule to  increase time by \emph{one time unit} in each
application:

\begin{maude}
crl [tickOne] : M$\;$:$\;$CLOCKS$\;$:$\;$NET => M$\;$:$\;$increaseClocks(M,$\,$CLOCKS,$\,$NET,$\,\highlight{1}$)$\;$:$\;$NET
                if  $\highlight{1}$ <= mte(M, CLOCKS, NET) .
\end{maude}

Analysis with such time sampling  is in general not sound and
complete, since it does not cover all possible system
behaviors: for example, if some transition's firing interval is
$[0.5,0.6]$, we could not execute that transition with this time
sampling.  Nevertheless, if all interval bounds are natural
numbers, then  ``all behaviors'' should be covered.

We
can therefore  quickly prototype our specification and experiment with
different parameter values, before applying the sound and complete
symbolic analysis and parameter synthesis methods developed in the
following sections.

The   term \code{net3($a$,$b$)} represents (a  more
general version of) the 
 net in Fig.~\ref{fig:producers}, where $a$ and $b$ are the
lower and upper bounds of the interval for transition $t_3$:

\begin{maude}
op net3 : Time TimeInf -> Net .
var LOWER : Time .  var UPPER : TimeInf .
eq net3(LOWER, UPPER)
 = "t1" : "p5" |-> 1 --> "p1" |-> 1 in [2 : 6] ;
   "t2" : "p1" |-> 1 --> "p2" |-> 1 ; "p5" |-> 1 in [2 : 4] ;
   "t3" : "p2" |-> 1 ; "p4" |-> 1 --> "p3" |-> 1 in [LOWER : UPPER] ;
   "t4" : "p3" |-> 1 --> "p4" |-> 1 in [0 : 0] .
\end{maude}

\noindent The initial marking in Fig.~\ref{fig:producers} is represented by the term
\texttt{init3}:

\begin{maude}
op init3 : -> Marking .
eq init3 = "p1"$\:$|->$\;$0 ; "p2"$\:$|->$\;$0 ; "p3"$\:$|->$\;$0 ; "p4"$\:$|->$\;$1 ; "p5"$\:$|->$\;$1 .
\end{maude}

We can simulate 2000 steps of the net with  different parameter
values:\footnote{Parts of Maude code and output
  from Maude executions will be replaced by `\texttt{...}' throughout
  the paper.}

\small
\begin{alltt}
Maude> \emph{\textcolor{blue}{rew} [2000] init3 : initClocks(net3(3,5)) : net3(3,5) .}

result State:
"p1"\(\;\)|->\(\;\)0 ; "p2"\(\;\)|->\(\;\)1 ; "p3"\(\;\)|->\(\;\)0 ; "p4"\(\;\)|->\(\;\)1 ; "p5"\(\;\)|->\(\;\)1 :  ...  :  ...
\end{alltt}
\normalsize

\noindent To further analyze the system, we define a function \code{k-safe}, where
\code{k-safe($n$,$\,m$)}   holds iff the marking $m$ does not have any
place with more than $n$ tokens:

\begin{maude}
  op k-safe : Nat Marking -> Bool .

  var M : Marking .  vars N N1 N2 : Nat .  var P : Place .
  eq k-safe(N, empty) = true .
  eq k-safe(N1, P$\:$|->$\;$N2 ; M) = N2 <= N1 and k-safe(N1, M) .
\end{maude}

We can then quickly (in 5ms) check whether the net is 1-safe when transition
$t_3$ has interval $[3,4]$:

\small
\begin{alltt}
Maude> \emph{\textcolor{blue}{search} [1] init3 : initClocks(net3(3,4)) : net3(\textcolor{red}{3,4})  =>*}
                  \emph{M : CLOCKS : NET} \emph{\textcolor{blue}{such that} not k-safe(\textcolor{red}{1}, M) .}

Solution 1 (state 27)
M --> "p1" |-> 0 ; \textcolor{red}{"p2" |-> 2} ; "p3" |-> 0 ; "p4" |-> 1 ; "p5" |-> 1
CLOCKS --> "t1" -> 0 ; "t2" -> 0 ; "t3" -> 4 ; "t4" -> 0
NET --> ...
\end{alltt}
\normalsize

\noindent The net is not 1-safe: we reached a state with two tokens in place
$p_2$.
However, the net is 1-safe if $t_3$'s interval is instead
$[2,3]$:

\small
\begin{alltt}
Maude> \emph{\textcolor{blue}{search} [1] init3 : initClocks(net3(2,3)) : net3(\textcolor{red}{2,3})  =>*}
                  \emph{M : CLOCKS : NET} \emph{\textcolor{blue}{such that} not k-safe(1,\,M) .}

No solution.
\end{alltt}
\normalsize

\noindent Further analysis shows that \texttt{net3(3,4)} is 2-safe, but that
\texttt{net3(3,5)}
is not even 1000-safe.

We can also analyze concrete  instantiations of our net by full linear
temporal logic (LTL) model checking in Maude. For example, we can
define a parametric atomic proposition
\code{place$\;p\;$has$\;n\;$tokens}, which holds in a state iff its
marking has exactly $n$ tokens in place $p$:

\begin{maude}
op place_has_tokens : Place Nat -> Prop [ctor] .
eq  (P$\:$|->$\;$N1 ; M : CLOCKS : NET) |= place P has N2 tokens = (N1 == N2) .
\end{maude}

Then we can check properties such as whether in \emph{each} behavior of the
system, there will be infinitely many states where $p_3$ has no tokens
\emph{and} infinitely many states where it holds one
token:\footnote{\texttt{[]}, \texttt{<>}, \texttt{/\char92}, and
  \texttt{\char126} are the Maude representations of corresponding (temporal)
  logic operators $\Box$ (``always''), $\Diamond$ (``eventually''),
  conjunction, and negation.}

\small
\begin{alltt}
Maude> \emph{\textcolor{blue}{red} modelCheck(init3 : initClocks(net3(3,4)) : net3(3,4),}
        \emph{(\textcolor{brown}{[] <>} place "p3" has 0 tokens) \textcolor{brown}{/\char92} (\textcolor{brown}{[] <>} place "p3" has 1 tokens)) .}

result Bool: true
\end{alltt}
\normalsize

We know that  \texttt{net3(3,4)} can reach markings with two
tokens in $p_2$; but is this inevitable (i.e., does it happen in \emph{all}
behaviors)?

\small
\begin{alltt}
Maude> \emph{\textcolor{blue}{red} modelCheck(init3 : initClocks(net3(3,4)) : net3(3,4),}
                      \emph{\textcolor{brown}{<>} place "p2" has 2 tokens) .}

result ModelCheckResult: counterexample(...)
\end{alltt}
\normalsize

\noindent The result is a counterexample showing a path where $p_2$ never holds
two tokens.

We also obtain a ``time sampling'' specification corresponding to
$\mathcal{R}_3$ by adding a global time
component to the state:

\begin{maude}
op _:_:_$\highlight{@_}$ : Marking ClockValues Net $\highlight{Time}$ -> State [ctor] .
\end{maude}

\noindent and modifying the tick rule to increase this global clock
according to the elapsed time. Furthermore, for time-bounded analysis
  we add  a constraint ensuring that system execution does not go
  beyond the time bound $\Delta$:

\small
\begin{alltt}
\textcolor{blue}{crl} \textcolor{ForestGreen}{[executableTick]} :
    M\(\;\):\(\;\)CLOCKS\(\;\):\(\;\)NET @ \textcolor{red}{GT} \(\:\)=> \(\:\)M\(\;\):\(\;\)increaseClocks(M,\,CLOCKS,\,NET,\,1)\(\;\):\(\;\)NET @ \textcolor{red}{GT\(\;\)+\(\;\)1}
    \textcolor{blue}{if}  GT < \(\Delta\)  and   \textcolor{ForestGreen}{--- remove this condition for unbounded analysis}
        1 <= mte(M, FT, NET) .
\end{alltt}
\normalsize

By setting $\Delta$ to \texttt{1000}, we can simulate one behavior of
the system \texttt{net3(3,5)} up to time 1000:

\small
\begin{alltt}
Maude> \emph{\textcolor{blue}{rew} init3 : initClocks(net3(3,5)) : net3(3,5) \textcolor{red}{@ 0} .}

result State:
"p1"\(\;\)|->\(\;\)0 ; "p2"\(\;\)|->\(\;\)1 ; "p3"\(\;\)|->\(\;\)0 ; "p4"\(\;\)|->\(\;\)1 ; "p5"\(\;\)|->\(\;\)1 : ... : ... \textcolor{red}{@ 1000}
\end{alltt}
\normalsize

We can then check whether \texttt{net3(3,4)} is one-safe in the time
interval $[5,10]$ by setting $\Delta $ in the tick rule to
\texttt{10}, and execute following command:

\small
\begin{alltt}
Maude> \emph{\textcolor{blue}{search} [1] init3 : initClocks(net3(3,4)) : net3(3,4) \textcolor{red}{@ 0} =>*}
            \emph{M : CLOCKS : NET \textcolor{red}{@ GT}} \emph{\textcolor{blue}{such that} not k-safe(1, M) \textcolor{red}{and GT >= 5} .}

Solution 1 (state 68)
MARKING --> "p1"\(\;\)|->\(\;\)0 ; \textcolor{red}{"p2"\(\;\)|->\(\;\)2} ; "p3"\(\;\)|->\(\;\)0 ; "p4"\(\;\)|->\(\;\)1 ; "p5"\(\;\)|->\(\;\)1
...
\textcolor{red}{GT --> 8}
\end{alltt}
\normalsize

\noindent This shows that the  non-one-safe marking can be reached in eight
time units.

\section{Parameters and Symbolic Executions}
\label{sec:sym}

Standard explicit-state Maude analysis of the theories
$\rtheorySem$--$\rtheorySemN{2}$ cannot be used to analyze all
possible
behaviors of PITPNs for two reasons:
\textbf{(1)}
 The rule \code{tick} introduces a new variable \code{T} in its right-hand
    side, reflecting the fact that time can
    advance by \emph{any} value  \code{T <= mte(...)}; and \textbf{(2)}
    analyzing \emph{parametric} nets with \emph{uninitialized}
  parameters is
  impossible with  explicit-state Maude analysis of
  concrete states. (For example, the condition \code{T in INTERVAL} in
  rule \texttt{applyTransition} will   never evaluate to \code{true} if \code{INTERVAL}
  is not a \emph{concrete} interval, and hence the rule will never
  be applied.)

Maude-SE analysis of \emph{symbolic} states with SMT variables can
 solve both  issues, by 
symbolically representing the time advances \code{T} and the net's
uninitialized parameters. This enables analysis and parameter synthesis methods  for  analyzing \emph{all}
possible behaviors in 
dense-time systems with unknown parameters.  

This section defines a  rewrite
theory $\rtheorySym$ that faithfully models PITPNs and that can be
symbolically executed using Maude-SE.
We prove that (concrete)
executions in $\rtheorySemN{1}$ are captured by (symbolic) executions in
$\rtheorySym$,  and vice versa.
We also show that standard folding
techniques \cite{DBLP:journals/jlap/Meseguer20} in rewriting modulo SMT are not
sufficient  for collapsing equivalent symbolic states in
$\rtheorySym$. We therefore
 propose a new folding technique that guarantees termination of the
 reachability 
analyses of $\rtheorySym$ when the
 state-class graph of the encoded PITPN is finite.

\subsection{The Symbolic Rewriting Logic
  Semantics}\label{subsec:sym-theory}

We define the ``symbolic'' semantics of PITPNs using the rewrite
theory $\rtheorySym$, which  is the symbolic counterpart
of $\rtheorySemN{1}$, instead of basing it on $\rtheorySemN{0}$,  since a symbolic
``tick'' step represents all 
possible tick steps from a  symbolic state. We therefore  do not
 introduce deadlocks not 
possible in the corresponding PITPN.   

$\rtheorySym$ is obtained from $\rtheorySemN{1}$
  by replacing
the sort \code{Nat} in markings  and the sort \code{PosRat} for clock values
with the corresponding SMT sorts \code{Integer} and
\code{Real}. (The former is only needed to enable reasoning
with \emph{symbolic} initial states where the number of tokens in a
location is unknown). Moreover, conditions in rules (e.g., \code{M1 <= M2})  are replaced
with the corresponding 
SMT expressions of sort \code{Boolean}. The symbolic execution of
$\rtheorySym$ in Maude-SE will 
accumulate and check the satisfiability of the constraints needed for a
parametric transition to happen.

We start by declaring the sort \code{Time} as follows:

\begin{maude}
sorts Time TimeInf .  subsort $\highlight{Real}$ < Time < TimeInf .
op inf : -> TimeInf [ctor] .
\end{maude}

\noindent where \code{Real} is the sort for SMT
reals. (We add constraints to the 
rewrite rules to  guarantee that only non-negative real numbers
are considered as time values.)

Intervals are defined as in $\rtheorySem$:
\lstinline{op `[_:_`] : Time TimeInf -> Interval}.
Since \lstinline{Real} is a subsort of \lstinline{Time},
an interval in $\rtheorySym$
may contain SMT variables. This means that a parametric interval
$[a,b]$ in a PITPN can be represented as the term
\code{[a:Real$\;$:$\;$b:Real]}, where
\code{a} and \code{b} are variables of sort \code{Real}.

The definition of markings, nets,  and clock values
is similar to the one in \Cref{subsec:eq-theory}. We only
need to adjust the following definition for markings: 

\begin{maude}
op _|->_ : Place $\highlight{Integer}$ -> Marking [ctor] .
\end{maude}

\noindent Hence, in a pair \code{$\nameit(p)$ |-> $e_I$}, $e_I$ is an SMT
integer expression that could be/include SMT variable(s).

Operations on markings and intervals remain the same,  albeit with the
appropriate SMT sorts. Since the operators in Maude for \code{Nat} and
\code{Rat} have the same signature that those for \code{Integer} and
\code{Real}, the specification needs few adjustments.
For instance,
the new definition of \code{M1 <= M2} is:

\begin{maude}
vars N1 N2 : $\highlight{Integer}$ .
op _<=_ : Marking Marking -> $\highlight{Boolean}$ .
eq ((P |-> N1) ; M1) <= ((P |-> N2) ; M2) = N1 $\highlight{<=}$ N2 and (M1 <= M2) .
eq empty <= M2 = true .
\end{maude}

\noindent where  \code{<=}   in \code{N1 <= N2} is a
function
\lstinline{op _<=_ : Integer Integer -> Boolean}.

Symbolic states in $\rtheorySym$ are defined as follows:

\begin{maude}
sort State. op _:_:_:_ : TickState Marking ClockValues Net -> State [ctor]
\end{maude}

The rewrite rules in $\rtheorySym$ act on
symbolic states that may contain SMT variables. 
Although  these rules  are similar to those in
$\rtheorySemN{1}$,  their symbolic execution  is completely
different. Recall from \Cref{sec:prelim}
that Maude-SE defines a theory transformation
to implement symbolic rewriting. In the 
resulting theory $\widehat{\rtheorySym}$,
when a rule is 
applied, the  variables occurring in the right-hand side but not in the
left-hand side are replaced by fresh variables. Moreover, 
rules in $\widehat{\rtheorySym}$ act on constrained terms of the form
$\phi\parallel t$, where $t$ in this case is a term of sort \code{State} and
$\phi$ is a satisfiable SMT boolean expression. The constraint $\phi$ is
obtained by accumulating the conditions in rules, thereby restricting the
possible values of the variables in $t$.

The tick rewrite rule in $\rtheorySym$ is

\begin{maude}
crl [tick] :  tickOk    : M : CLOCKS     : NET
          =>  tickNotOk : M : increaseClocks(M, CLOCKS, NET, $\highlight{T}$) : NET
if ($\highlight{T >= 0/1}$ and mte(M, CLOCKS, NET, T)) .
\end{maude}

The variable \code{T} is restricted to be a non-negative real
number and to satisfy the following \emph{predicate} \code{mte}, which gathers
the constraints to ensure that time cannot advance beyond the point in
time when an enabled transition \emph{must} fire:

\begin{maude}
op mte : Marking ClockValues Net $\highlight{Real}$ -> $\highlight{Boolean}$ .
eq mte(M, empty, NET, T) = true .
eq mte(M, (L -> R1) ; CLOCKS, (L : PRE --> ... in [T1 : $\highlight{inf}$]) ; NET , T) 
 = mte(M, CLOCKS, NET, T) .
eq mte(M, (L -> $\highlight{R1}$) ; CLOCKS, (L : PRE --> ...  in [T1 : $\highlight{T2]}$) ; NET, $\highlight{T}$) 
 = active(M, L : ...) $\highlight{? T <= T2 - R1 : true}$) and mte(M, CLOCKS, NET, T) .
\end{maude}

This means that, for every transition \code{L},  if the upper bound of the
interval in \code{L} is \code{inf}, no restriction on \code{T} is added.
Otherwise, if \code{L} is active at marking \code{M}, the SMT ternary operator
\code{C$\;$?$\;$E1$\;$:$\;$E2} (checking \code{C} to choose either
\code{E1} or \code{E2}) 
further constrains \code{T} to be less than \code{T2$\,$-$\,$R1}.  The
definition of \code{increaseClocks} also uses  this SMT 
operator to represent the 
new values of the clocks:

\begin{maude}
eq increaseClocks(M, (L -> R1) ; CLOCKS, (L : PRE --> ... ) ; NET, T)
 = (L -> (active(M, L : PRE ...) $\highlight{? R1 + T : R1 }$)) ;
   increaseClocks(M, CLOCKS, NET, T) .
\end{maude}

The rule for applying a transition is defined as follows:

\begin{maude}
crl [applyTransition] :
    TS : M : ((L -> T) ; CLOCKS) : (L : PRE --> ...) ; NET)
 => tickOk : ((M - PRE) + POST) :  updateClocks(...) :
    (L : PRE --> ... ; NET) if active(...) and (T in INTERVAL) .
\end{maude}

When applied, this rule adds new constraints asserting that
the transition \code{L} can be fired
(predicates \code{active} and \code{_in_}) and  updates
the state of the clocks:
\begin{maude}
eq updateClocks((L' -> R1)$\;\,$;$\;\,$CLOCKS, INTERM-M, (L'$\,$:$\,$PRE --> ...); NET) 
 = (L ->  PRE <= INTERM-M $\highlight{? R1 : 0/1}$) ; updateClocks(...) .
\end{maude}

\begin{figure}
    \begin{center}
    \includegraphics[width=0.7\textwidth]{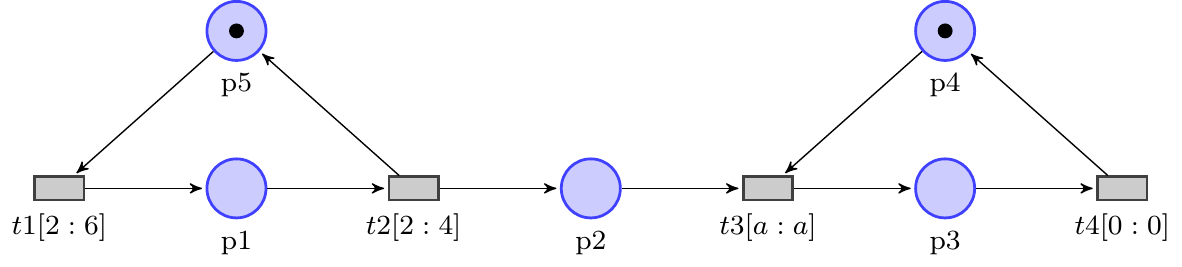}
\end{center}
    \caption{A simple production-consumption system taken from
      \cite{Wang1998}.\label{fig:producers}} 
\end{figure}

In the following, \code{k-safe($k$,$m$)} is a predicate stating that
the marking $m$ does not have more than $k$ tokens  
in any place.

\begin{example}\label{ex:producers}
Let \code{$\mathit{net}$}
and \code{$m_0$} 
    be the Maude terms representing, respectively,   the PITPN and  the initial
    marking shown in \Cref{fig:producers}. The term \code{$\mathit{net}$}
    includes a variable \code{a:Real} representing the parameter
    $a$. The following command 

\begin{maude}
smt-search tickOk : $m_0$ : initClocks($\mathit{net}$) : $\mathit{net}$  =>*  TICK : M : CLOCKS : NET
  such that ($\highlight{a:Real >= 0/1}$ and $\highlight{not k-safe(1, M)}$) = true .
\end{maude}
answers the question whether it is possible to reach a state with a marking
    $M$ with more than one token in some place. Maude positively answers this
    question and the resulting accumulated constraint tells us that
    such a state
    is reachable (with 2 tokens in $p_2$) if  \code{a:Real >= 4/1}.
\end{example}

Terms of sort \code{Marking} in $\rtheorySymN{1}$ may contain expressions
with parameters (i.e., variables)  of sort \code{Integer}. 
Let  $\Lambda_m$  denote
the set of such parameters and $\pi_m:\Lambda_M \to \grandn$ 
a valuation function for them. 
We use $m_s$ to denote a mapping from places to \code{Integer}
expressions including parameter variables. Similarly, $\mathit{clocks}_s$
 denotes a mapping from transitions
to \code{Real} expressions (including variables). 
We write 
$\pi_m(\mathit{m}_s)$ 
to denote the ground term where the parameters in
markings are replaced by the corresponding values
$\pi_m(\lambda_i)$. Similarly for 
$\pi(\mathit{clocks}_s)$ .
We use $\encSym{\PN}$ to
  denotes the above rewriting logic 
  representation of nets in  $\rtheorySymN{1}$. 
  
  Let $t_s$ be a term  of 
  sort \code{State} in $\rtheorySymN{1}$
  and assume that 
  $\phi \parallel t_s
  \rightsquigarrow_{\rtheorySymN{1}} \phi' \parallel t'_s$.
  By construction, if for all $t\in \os \phi\parallel t_s \cs$ 
  all markings (sort \code{Integer}), clocks and parameters (sort \code{Real}) are non-negative 
  numbers, then this is also the case for all $t' \in \os \phi' \parallel t_s'\cs$.
  Note that there is a one-to-one correspondence for ground terms in $\rtheorySymN{1}$ (sorts \code{Marking}, 
  \code{ClockValues}, etc) satisfying that condition
  with (ground) terms in $\rtheorySemN{1}$. We use 
   $t\approx \in\os \phi \parallel t_s\cs$ to denote that 
   there exists a $\rtheorySymN{1}$ term $t'\in  \os \phi \parallel t_s\cs$
   and $t$ is its corresponding term in $\rtheorySemN{1}$.

  The following theorem states that the symbolic semantics matches all the behaviors
resulting from a  concrete execution of $\rtheorySemN{1}$ with
arbitrary parameter valuations $\pi$ and $\pi_m$. Furthermore,  
 for all symbolic executions with parameters, there exists
a corresponding concrete execution where the parameters are instantiated with 
values consistent with the resulting accumulated constraint.

\begin{theorem}[Soundness and Completeness]\label{th:sym-correct}
    Let $\PN$ be a  PITPN
    and $m_s$ be a marking possibly including parameters. 

\noindent\textbf{(1)}
            Let $\phi$ be the constraint 
                            $\bigwedge_{\lambda_i\in \Lambda}(0 \leq
                            \lambda_i^- \leq \lambda_i^+) 
                            \wedge \bigwedge_{\lambda_i\in \Lambda_m}(0 \leq \lambda_i)
                            $. If

                            $\phi\parallel
                            \mathtt{tickOk}\;\code{:}\;m_s\;\code{:}\;\mathit{clocks_s}\;\code{:}\;\encSym{\PN}
                            \rightsquigarrow^*_{\rtheorySymN{1}} 
                            \phi'\parallel TS'\;\code{:}\;m'_s\;\code{:}\;\mathit{clocks'_s}\;\code{:}\;\encSym{\PN}
                            $
                            then, there exists $\pi$ and  $\pi_m$ s.t. 
                            $\mathtt{tickOk}\:\code{:}\;\pi_m(m_s)\;\code{:}\;\mathit{clocks}\;\code{:}\;\encBase{\pi(\PN)}
                            \longrightarrow^*_{\rtheorySemN{1}} 
                            TS'\:\code{:}\;\pi_m(m'_s)\;\code{:}\;\mathit{clocks'}\;\code{:}\;\encBase{\pi(\PN)}
                            $ where 
                            $\phi'\wedge\bigwedge_{\lambda_i\in\Lambda}
                            {\lambda_i} = \pi(\lambda_i) 
                            \wedge \bigwedge_{\lambda_i\in\Lambda_m}
                            {\lambda_i} = \pi_m(\lambda_i)$ 
                            is satisfiable,
                            $\mathit{clocks}\approx\in \os \phi \parallel \mathit{clocks}_s \cs$
                            and $\mathit{clocks'}\approx\in \os \phi' \parallel \mathit{clocks'}_s \cs$.

                            \noindent\textbf{(2)}
            Let $\pi$ be a parameter valuation and $\pi_m$ a parameter marking valuation. 
            Let $\phi$ be the constraint 
                            $\bigwedge_{\lambda_i\in \Lambda}(\lambda_i=\pi(\lambda_i))
                            \wedge \bigwedge_{\lambda_i\in
                              \Lambda_m}(\lambda_i=\pi_m(\lambda_i)) 
                            $.
            If 

            $\mathtt{tickOk}\;\code{:}\;\pi_m(m_s)\;\code{:}\;\mathit{clocks}\;\code{:}\;\encBase{\pi(\PN)}
              \longrightarrow^*_{\rtheorySemN{1}} 
            TS'\;\code{:}\;m'\;\code{:}\;\mathit{clocks'}\;\code{:}\;\encBase{\pi(\PN)}$,
                 then 

                   $\phi \parallel \mathtt{tickOk}\;\code{:}\;m_s\;\code{:}\;\mathit{clocks_s}\;\code{:}\;\encSym{\PN}
                  \longrightarrow^*_{\rtheorySymN{1}} 
                   \phi' \parallel TS'\;\code{:}\;m_s'\;\code{:}\;\mathit{clocks_s}'\;\code{:}\;\encSym{\PN}$
                            where 
                            $m' \approx\in \os \phi'\parallel m'_s\cs$, 
                            $\mathit{clocks}\approx\in \os \phi \parallel \mathit{clocks}_s \cs$
                            and $\mathit{clocks'}\approx\in \os \phi' \parallel \mathit{clocks'}_s \cs$.
\end{theorem}

The symbolic counterpart $\rtheorySymN{2}$ of the theory $\rtheorySemN{2}$
can be defined similarly.

\subsection{A New Folding Method for Symbolic
  Reachability}\label{subsec-folding}

Reachability analysis should terminate 
for both
positive and negative queries for nets with   
finite parametric state-class graphs.
However, the symbolic state  space generated by 
\lstinline{smt-search} is infinite even for such nets, so that
 \lstinline{smt-search} will not terminate when the desired states are
 unreachable. The problem is that \lstinline{smt-search} stops exploring from  a
symbolic state only if it has already visited  the \emph{same}
state. Due to the fresh variables created 
in $\rtheorySym$ whenever the \texttt{tick} rule is applied, symbolic
states representing the same set 
of concrete states are not the same, even though they are \emph{logically}
equivalent, as exemplified below.

\begin{example}\label{ex:infinite}
    The following command, trying to show that the PITPN in
    \Cref{fig:producers} is 1-safe if  $0 \leq a<4$, does not 
    terminate. 
\begin{maude}
smt-search tickOk$\;$:$\;m_0\;$:$\;$0-clock($net$)$\;$:$\;\mathit{net}$  =>*  TICK : M : CLOCKS : NET 
such that ($\highlight{a:Real >= 0/1 and a:Real < 4}$ and $\highlight{not M <= k-safe(1,M)}$) = true$\,$.
\end{maude}

\noindent Furthermore, the command 

\begin{maude}[escapechar=!]
smt-search tickOk$\;$:$\;m_0\;$:$\;$0-clock($net$)$\;$:$\;\mathit{net}$  =>*  TICK : M : CLOCKS : NET 
such that ($\highlight{a:Real >= 0/1 and a:Real < 4}$ and !!\highlight{ M <=  $m_0$ and $ m_0$ <= M}) = true$\,.$ 
\end{maude}

\noindent searching for reachable states where 
    $M = m_0$ will produce infinitely many (equivalent) solutions,
    including, e.g., the following constraints:  

\noindent\adjustbox{width=\textwidth}{
    \begin{maude}
Solution 1:  #p5-9:Integer === 1 and #t3-9:Real  + a:Real -  #t2-9:Real <= 0/1 and ...
Solution 2: #p5-16:Integer === 1 and #t3-16:Real + a:Real - #t2-16:Real <= 0/1 and ...
    \end{maude}
}

\noindent     where a variable created by \lstinline{smt-search} starts with \verb@#@
    and ends with a number taken from a sequence to guarantee freshness. 
    Let $\phi_1\parallel t_1$ and $\phi _2\parallel t_2$ be, respectively,  the 
    constrained terms found in \code{Solution 1} and \code{Solution 2}. 
    In this particular output, $\phi_2\parallel t_2$ is obtained by further rewriting 
    $\phi_1\parallel t_1$. 
    The variables representing the state of markings and clocks 
    (e.g.,  \code{\#p5-9} in $t_1$ and \code{\#p5-16} in $t_2$) are clearly different,
    although they 
    represent the same set of concrete values ($\os \phi_1\parallel
    t_1\cs = \os \phi_2\parallel t_2 \cs$).  
    Since constrains are accumulated when a rule is applied, we note that 
    $\phi_2$ equals $\phi_1 \wedge \phi_2'$ for some $\phi_2'$, and 
    ${\it vars}(\phi_1\parallel t_1)\subseteq {\it vars}(\phi_2\parallel t_2)$. 
\end{example}

The usual approach for collapsing equivalent symbolic states in rewriting 
modulo SMT is subsumption
\cite{DBLP:journals/jlap/Meseguer20}. Essentially, we stop searching 
from a symbolic state if, during the search, we have already encountered
another symbolic state that subsumes (``contains'') it.
More precisely, let $U = \phi_u \parallel t_u$ and
$V = \phi_v \parallel t_v$ be constrained terms. 
Then 
$U  \sqsubseteq  V$  if there is
a substitution $\sigma$ such that $t_u = t_v\sigma$  and
the implication $\phi_u \Rightarrow \phi_v\sigma $ holds. In that case,
$\llbracket U \rrbracket \subseteq \llbracket V \rrbracket$.
A search will not further explore a constrained term $U$ if another
constrained term $V$ with $U \sqsubseteq V$ has already been
encountered. It is known that such reachability analysis with folding
is sound (does not 
generate spurious counterexamples~\cite{bae2013abstract}) 
but not necessarily complete (since
$\llbracket U \rrbracket \subseteq \llbracket V \rrbracket$
does not imply $U \sqsubseteq V$). 

\begin{example}
    Let $\phi_1$ and $\phi_2$ be
    the resulting constraints in the two solutions found by the second \lstinline{smt-search} command in \Cref{ex:infinite}. 
    Let $\sigma$ be the substitution that maps \code{\#p$i$-9} to \code{\#p$i$-16}
    and \code{\#t$j$-9}  to \code{\#t$j$-16} for each place $p_i$ and
    transition $t_j$.  
The SMT solver determines that  the formula $\neg (\phi_2
    \Rightarrow \phi_1\sigma)$ 
    is satisfiable (and therefore  $\phi_2 \Rightarrow \phi_1\sigma$ is not
    valid).
    Hence, a procedure based on checking this implication will fail
    to determine that the state in the second solution can be subsumed by the state
    found in the first solution.

\end{example}

The satisfiability witnesses of $\neg (\phi_2 \Rightarrow \phi_1\sigma)$ can
give us some ideas on how to make the subsumption procedure more precise.
Assume that $\phi_1$ carries the information $R = T_0$ for some
clock represented by $R$ and $T_0$ is a tick variable subject to $\phi = (0 \leq
T_0 \leq 2)$. Assume also that in $\phi_2$, the value of the same clock is 
$R' = T_1 + T_2$ subject to $\phi' = (\phi\wedge T_1 \geq 0 \wedge T_2\geq 0 \wedge T_1+T_2 \leq 2)$.
Let $\sigma= \{R \mapsto R'\}$. Note that  $(R'=T_1+T_2
\wedge \phi \wedge \phi')$ does not imply $(R = T_0\wedge \phi)\sigma$ (take, e.g., 
the valuation $T_1=T_2=0.5$ and $T_0=2$). The key
observation is that, even if $R$ and $R'$ are both constrained to be
in the interval $[0,2]$ (and hence represent the same state for this
clock), the assignment of $R'$ in the antecedent does not need 
to coincide with the one for $R$ in the consequent of the implication. 

In the following,  we propose a subsumption relation that solves the
aforementioned problems.  
Let $\phi \parallel t$ be a constrained term where $t$ is a term of sort \code{State}. 
Consider the abstraction of built-ins $(t^\circ, \sigma^\circ)$ for
$t$, 
where $t^\circ$ is as $t$ but it replaces the expression $e_i$ in
markings ($p_i\mapsto e_i$)  
and clocks ($l_i \to e_i$)  with new fresh variables.
The substitution $\sigma^\circ$ is defined accordingly. 
Let $\Psi_{\sigma^\circ} = \bigwedge_{x \in
  \mathit{dom}(\sigma^\circ)} x = x \sigma^\circ $. 
We use $\projnow{(\phi\parallel t)}{}$ to denote the constrained 
term $\phi \wedge \Psi_{\sigma^\circ} \parallel t^\circ $. 
Intuitively,  $\projnow{(\phi\parallel t)}{}$ replaces the clock values and markings
 with fresh variables and the 
 boolean expression $\Psi_{\sigma^\circ}$
 constrains those variables to take the values
 of clocks and marking in $t$. From \cite{rocha-rewsmtjlamp-2017} we can show that 
 $\os \phi \parallel t \cs = \os \projnow{(\phi\parallel t)}{}\cs$.

Note that the 
only variables occurring in $\projnow{(\phi\parallel t)}{}$
are those for parameters (if any) 
and the fresh variables in $\mathit{dom}(\sigma^\circ)$
(representing the symbolic state of clocks and markings). 
For a constrained term $\phi\parallel t$, we use
$\exists(\phi \parallel t)$ to denote the formula
$(\exists X) \phi$ where 
$X = \mathit{vars}(\phi) \setminus \mathit{vars}(t)$.

 \begin{definition}[Relation $\preceq$]\label{def:rel}
 Let $U = \phi_u \parallel t_u$ and $V = \phi_v \parallel t_v$ be 
 constrained terms where $t_u$ and $t_v$ are terms of sort \code{State}. 
 Moreover, let $\projnow{U}{} = \phi_u'\parallel t_u'$ and 
$\projnow{V}{}= \phi_v'\parallel t_v'$,
where $\ovars{t_u'} \cap \ovars{t_v'} = \emptyset$.
We define 
 the relation $\preceq$ on constrained terms 
 so that  $U\preceq V$ 
 whenever there exists a substitution $\sigma$ such that  $t_u' = t_v'\sigma$
 and the formula 
 $\exists(\projnow{U}{}) \Rightarrow \exists(\projnow{V}{})\sigma$ is valid. 
\end{definition}

The formula $\exists(\projnow{U}{})$ hides the information about 
 all the tick variables as well as the information about the clocks and markings
 in previous time instants. What we obtain is the information about
 the parameters and the values of the clocks and markings ``now''. 
Moreover, if $t_u$ and $t_v$ above are both \code{tickOk} states (or both
 \code{tickNotOk} states),
  and they represent two symbolic states of the same PITPN, 
 then $t_u'$ and $t_v'$ always match ($\sigma$ being  
 the identity on the variables representing parameters
 and mapping the corresponding variables created
 in $\projnow{V}{}$  and $\projnow{U}{}$).

 \begin{theorem}[Soundness and Completeness]\label{th:folding}
     Let $U$  and $V$  be
     constrained terms in $\widehat{\rtheorySym}$ representing two symbolic states of the same
     PITPN.
Then,
     $\os U \cs \subseteq \os V \cs$ iff $U \preceq V$. 
\end{theorem}

 We have implemented a new symbolic reachability analysis 
 based on the folding relation in \Cref{def:rel}.
 Building on the theory transformation defined in Maude-SE,
we transform 
 the theory $\rtheorySymN{1}$ into a rewrite theory
$\rtheorySymNF{1}$
 that rewrites terms of the form 
 $S :\phi\parallel t$ where $S$ is a set of constrained terms  (the 
 already visited states).
 Theory $\rtheorySymNF{1}$ defines the sort \code{SetState}
 for \code{;}-separated sets of constrained terms 
 and an operator 
 \code{subsumed}$(\phi\parallel t~,~ S)$ that 
 reduces to \code{true} iff there exists $\phi'\parallel t' \in S$
 s.t $\phi\parallel t \preceq \phi' \parallel t'$. 
 A rule 
 $l : q \longrightarrow r \mbox{ \textbf{if} } \psi$
 in $\rtheorySymN{1}$
is transformed
into
the following rule in $\rtheorySymNF{1}$:
\begin{align*}
l:~
S : \mathtt{PHI} \parallel q^\circ
\longrightarrow
(S ; \phi_r\parallel r) : \phi_r \parallel r
\mbox{ \textbf{if} }
&
\mathtt{smtCheck}(\phi_r) \wedge \mathtt{not~subsumed}(\phi_r\parallel r, S)
\end{align*}

\noindent
where $\mathtt{PHI}$ is a \code{Boolean} variable,
$\mathtt{S}$ is a variable of sort \code{SetState}, 
$(q^\circ, \sigma^\circ)$ is an abstraction of built-ins for $q$ and 
$\phi_r = (\mathtt{PHI} \mathbin{and} \psi \mathbin{and} \Psi_{\sigma^\circ} )$.
Note that the transition happens only if the new state 
$\phi_r\parallel r$ is not subsumed by an already visited state in $S$. 
The theory $\rtheorySymNF{2}$ is similarly obtained
from $\rtheorySymN{2}$.

In $\rtheorySymNF{1}$, for an initial constraint $\phi$ on the parameters,
the command 

\noindent\lstinline[mathescape]{search [$n$,$m$] : empty:$\phi\parallel t$ =>* $S:\phi'\parallel t'$ such that smtCheck($\phi'\wedge \Phi)$}
answers the question whether it is possible to reach a symbolic 
state that matches $t'$ and satisfies the condition $\Phi$. 
In the following, we use $\code{init}(net, m_0, \phi)$ to denote the term 
$\code{empty}:\phi\parallel \mathtt{tickOk}\;\code{:}\;m_0\;\code{:}\code{initClocks}(net)\;\code{:}\;net$.

 \begin{example}\label{ex:finite}
Consider  the  PITPN in Fig. \ref{fig:producers}.
    Let  $m_0$ be the marking in the figure 
    and $\phi=0 \leq a < 4$. The command 

     \begin{maude}
search init($net$, $m_0$, $\phi$) =>* S : $\phi' \parallel $ ( TICK : M : CLOCKS : NET ) 
   such that smtCheck($\phi'$ and not k-safe(1,M)) . \end{maude}

     \noindent terminates returning \code{No solution}, showing that 
    the net is 1-safe if $0 \leq a < 4$.  
\end{example}

The following  result shows that if the set of reachable state classes in the symbolic
semantics of $\PN$ (see \cite{EAGPLP13}) is finite, then so  is the set of 
 reachable symbolic states using the new folding technique.

 \begin{corollary}\label{corollary:term}
For any PITPN $\PN$ and state class $(M,D)$,
if the transition system $(\mathcal{C}, (M,D), \Fleche{})$ is finite,
then 
so  is 
$\left(T_{\Sigma , \texttt{State}}, \code{init}(\PN,M,D), \symredif\right)$.
 \end{corollary}

It is worth noting that 
the new folding relation in Def.~\ref{def:rel} and Theorem~\ref{th:folding}
is applicable to any rewrite theory $\mathcal{R}$ that satisfies the requirements for 
rewriting with SMT~\cite{rocha-rewsmtjlamp-2017}, briefly explained in Sec.~\ref{sec:rew-smt}.

\section{Parameter Synthesis and Symbolic Model
  Checking}\label{sec:analysis}

This section  shows how Maude-SE
can be used for a wide
range of formal analyses beyond reachability analysis. We show how to use Maude-SE for solving
parameter synthesis problems,  model checking the classes of
non-nested timed temporal logic properties supported by the state-of-the-art
PITPN tool \romeo{},   reasoning with
parametric initial states where the number of tokens in the different
places is not known, and analyzing  nets with  user-defined execution
strategies. We thereby provide analysis methods
that go beyond those supported by \romeo{}, while
supporting almost all forms of analysis provided by \romeo{}.

\subsection{Parameter Synthesis}

A \emph{state predicate} is a boolean expression whose atomic propositions
include tests on the values of
markings (e.g., \code{k-safe(1,$\,m$)}) and clocks
(e.g., $c_1\, \code{<}\, c_2$).    \emph{\EF{}-synthesis} is the problem of computing
parameter values $\pi$ such that there exists a run of $\pi(\PN)$ that reaches a state
satisfying a given state predicate $\phi$. The \emph{safety synthesis problem}
\AG{}$\neg\phi$  is the problem of computing the  parameter values
for which states satisfying $\phi$ are unreachable.

\lstinline{search} in the theory  $\rtheorySymNF{1}$ (see
Section~\ref{subsec-folding}) provides
semi-decision procedures for
solving these parameter synthesis problems (which are
undecidable in general).
As illustrated below, the resulting constraint computed by
\lstinline{search}  can be used to 
synthesize the parameter values that allow such execution paths.
The safety synthesis problem \AG{}$\neg\phi$
  can be solved by finding all  solutions for \EF{}$\phi$
and then negating the resulting constraint.

\begin{example}\label{ex:analysis1}
  \Cref{ex:producers} shows an \EF-synthesis problem: find values
  for the parameter $a$ such that a state with at least two tokens in
  \emph{some} place can be reached. If $\phi = 0 \leq a$, the command
  \begin{maude}
search [1] init($net$, $m_0$, $\phi$) =>* S : PHI' $\parallel $ ( TICK : M : CLOCKS : NET ) 
   such that smtCheck(PHI' and not k-safe(1,M)) .
\end{maude}
returns one solution and the resulting constraint $\phi'$, instantiating 
the pattern \code{PHI'}, 
  can be used 
  to extract the parameter values as follows.
Let $X$ be the set
  of SMT variables in $\phi'$ \emph{not}  representing parameters.
A
  call to the quantifier elimination procedure
(\code{qe}) of the SMT solver Z3 on the formula $\exists X. \phi'$ reduces to
    $\code{a:Real >= 4/1}$, giving us the desired values for the
    parameter $a$.   \end{example}

  To solve the safety synthesis problem 
\AG{}$\neg\phi$, we have used Maude's 
   meta-programming facilities~\cite{maude-book} to 
  implement a command 
\lstinline[mathescape]{safety-syn($\mathit{net}$,$\,m_0$,$\,\phi_0$,$\,\phi$)}
where $m_0$ is a marking, $\phi_0$ a constraint on the parameters
and $\phi$ a constraint involving the variables \code{M} and \code{CLOCKS}
as in the \lstinline{search} command in Example~\ref{ex:analysis1}. This command 
iteratively calls  \lstinline{search}   to find a state reachable from
$m_0$, with initial constraint $\phi_0$, 
   where $\phi$ does not hold. If such state is found,
   with accumulated constraint $\phi'$,
   the \lstinline{search} command is invoked again with initial
   constraint $\phi_0 \wedge \neg \phi'$. This process stops
   when no more reachable states where $\phi$ does not hold are found,
   thus 
solving the
   \AG{}$\neg\phi$ synthesis
   problem.

  \begin{example}  \label{ex:analysis2}
    Consider the PITPN in Fig.~\ref{fig:ex-scheduling}, taken from
    \cite{DBLP:journals/jucs/TraonouezLR09}, with a parameter $a$ and three parametric transitions
    with intervals $[a:a]$, $[2a:2a],$ and $[3a:3a]$.
    \romeo{} can synthesize the values of the parameter $a$
    making the net 1-safe, subject to  initial constraint
    $30 \leq a \leq 70$.
    The same query can be answered
    in Maude:
    \begin{maude}
safety-syn($\mathit{net}$, $m_0$, a:Real >= 30/1 and a:Real <= 70/1, k-safe(1,M)) . \end{maude}
The first counterexample found assumes 
    that $a \leq 48$. If  $a>48$, \lstinline{search}
    does not find any  counterexample. This is the same answer
    that  \romeo{} found. \end{example}

\romeo{} only supports properties
over markings. The state predicates
in the commands above can include also conditions on
the clock values.

Our  symbolic theories  allow for parameters
(variables of sort \code{Integer})  in the initial marking. This opens up the
possibility of using Maude-SE to solve synthesis problems involving
parametric initial markings.  For instance, we can determine the initial markings that
make the net k-safe and/or alive:

\begin{example}\label{ex:pmarking}   Consider
     a \emph{parametric} initial marking $m_s$ for the net in
     \Cref{fig:producers},  with parameters $x_1$, $x_2$, and $x_3$
     denoting the number of tokens in
    places $p_1$, $p_2$,  and $p_3$,  respectively, and the initial
    constraint $\phi_0$ stating 
    that $a\geq 0$ and $ 0 \leq  x_i \leq 1$. The execution of the command
    \lstinline{safety-syn}($net, m_s, \phi_0, \texttt{k-safe(1,M)}$) 
    determines  that the net is 1-safe when $x_1=x_3=0$
    and $0\leq x_2 \leq 1$. 

\end{example}

\paragraph{Analysis with strategies.}
Maude's strategy facilities \cite{maude-manual} allow us to  analyze 
PITPNs whose executions  follow  some user-defined 
strategy. As exemplified below, such strategies may affect the 
outcome of parameter synthesis analysis. 
\begin{example}
We execute the net in  Fig.~\ref{fig:producers}
 with the following strategy \lstinline{t3-first}: 
 whenever transition $t_3$  and some other
 transition are enabled
 at the same time, then $t_3$  fires first. 
This execution strategy can be specified as follows: 
 \begin{maude}
t3-first := ($\;$applyTransition[$\,$L <- "t3"$\,$] or-else all )!\end{maude}
Running 
\lstinline[mathescape]{srew init($net$, $m_0$, $a\geq0$) using t3-first} in
 $\rtheorySymNF{1}$ finds all  symbolic  states  reachable with
 this strategy, and all of them are 1-safe. Therefore,  all parameter values $a
 \geq 0$ guarantee the desired property with this
 execution strategy. \end{example}
 \subsection{Analyzing Temporal Properties}\label{sec:model-checking}

This section shows how Maude-SE can be used to analyze the
 temporal properties supported by  \romeo{}~\cite{romeo},
 albeit in a few 
 cases without parametric bounds in the temporal formulas.  \romeo{} can analyze the following temporal properties:
\[
    \textbf{Q}\, \phi\,\CTLU_J\,\psi \;\mid\; \textbf{Q} \CTLF_J\,\phi \;\mid\;
\textbf{Q} \CTLG_J \,\phi \;\mid\; \phi \rightsquigarrow_{\leq b} \psi
\]where $\textbf{Q}\in \{\exists, \forall\}$ is the existential/universal path
quantifier, $\phi$ and $\psi$ are \emph{state predicates} on \emph{markings},
and $J$ is a time
interval  $[a,b]$, where $a$ and/or $b$ can be parameters and $b$ can
be $\infty$. For example,
$\forall  \CTLF_{[a,b]}\,\phi$ says that
in \emph{each} path from the initial state,  a marking satisfying
$\phi$ is  reachable in some time
in  $[a,b]$.
The bounded  response $\phi
\rightsquigarrow_{\leq b} \psi $ denotes the formula
$\forall\CTLG(\phi \,\Longrightarrow\, 
\forall\CTLF_{[0,b]} \,\psi)$ (each $\phi$-marking \emph{must} be followed
by a $\psi$-marking  within time $b$).  

Since queries include time bounds, we use the theory 
$\rtheorySymNF{2}$,  and 
$\code{init}(net, m_0, \phi)$ will denote the term
$\code{empty}:\phi\parallel \mathtt{tickOk}\;\code{:}\;m_0\;\code{:}\;\code{initClocks}(net)\;\code{:}\;net\;\mathtt{@\,0/1}$. 

State predicates, including 
inequalities on markings and clocks,
and also a test
whether the global clock is in a given interval
are defined as follows:
\begin{maude}
ops _>=_ _>_ _<_ _<=_ _==_ : Place Integer -> $\highlight{Prop}$ .
ops _>=_ _>_ _<_ _<=_ _==_ : Clock Real -> $\highlight{Prop}$ .
op in-time : Interval -> Prop .
eq S : C || (TICK : M ; (P |-> N1) : CLOCKS : NET) @ G-CLOCK  $\highlight{|=}$  P >= N1'
  = $\highlight{smtCheck}$(C and N1 >= N1' ) . --- similarly for >, <=, < and ==
eq S : C || (TICK : M : CLOCKS : NET) @ G-CLOCK  $\highlight{|=}$  in-time INTERVAL
  = $\highlight{smtCheck}$(C and (G-CLOCK in INTERVAL )) .
\end{maude}

Atomic propositions (sort \code{Prop}) are evaluated on symbolic states represented
as constrained terms 
 $S:\phi\parallel t$. Since they may contain variables, 
a call to the SMT is needed to determine whether the  constraint
$\phi$ entails the  proposition. 

Some of the temporal formulas supported by \romeo{} can be easily verified
using the reachability commands presented in the previous section.
The property
$\exists\CTLF_{[a,b]}\,\psi$ can be verified using the command:
\begin{maude}
search [1] init($net,m_0,\phi$) =>* S : PHI' $\parallel $  TICK : M : CLOCKS : NET @ G-CLOCK
   such that $\mathit{STATE'}$ |= $\psi$ and G-CLOCK in [$a\:$:$\:b$] .
\end{maude}
where $\phi$ states that all  parameters are non-negative  numbers
and $\mathit{STATE'}$ is the expression to  the right of \texttt{=>*}.
$a$ and $b$ can be variables representing parameters to be synthesized;
 and
 $\psi$ can be an expression involving \code{CLOCKS}.
 For example, 
\begin{maude}
search [1] init($net,m_0,\phi$) =>*
 S' : PHI' $\parallel$ TICK : (M ; "p1"$\:$|->$\;$P1) : (CLOCKS ; "t2"$\:$->$\;$C2) : NET @ G-CLOCK
 such that $\mathit{STATE'}$ |= P1$\;$>$\;$1 and C2$\;$<$\;$2/1 and G-CLOCK in [$a$ : $b$] .
\end{maude}
checks whether it is possible to reach a marking,   in some time in
$[a,b]$, with more than one token 
in place $p_1$, when the value of the clock of transition $t_2$ is $<2$.

The dual property  $\forall \CTLG_{[a,b]} \:\phi$ can be
checked by analyzing 
  $\exists \CTLF_{[a,b]} \:\neg\,\phi$.

\begin{example}
    Consider the PITPN in \Cref{ex:analysis2} with (interval) parameter
    $\phi = 30\leq a\leq 70$. The property
    $\exists\CTLF_{[b,b]}(\neg {  \it 1\mhyphen safe} )$
    can be verified with the following command,
which determines that the parameter $b$ satisfies $60 \leq b \leq 96$.
\begin{maude}
search [1] init($net,m_0,\phi$)  =>*  S : PHI' $\parallel $ TICK : M : CLOCKS : NET @ G-CLOCK
  such that $\mathit{STATE'}$ |= b:Real >= 0/1 and (G-CLOCK in [b:Real$\:$:$\:$b:Real]) 
                      and not (k-safe(1,M)) .
\end{maude}
  \end{example}

The bounded response $\phi \rightsquigarrow_{\leq b} \psi $ formula
can  be verified using a simple theory transformation on
$\rtheorySymN{0}$ followed by
reachability analysis. The theory transformation
adds a new constructor for the sort \code{State}
to build terms of the form $\red{C_\phi}\texttt{ : }M\texttt{ : }{\mathit Clocks}\texttt{ : } {\mathit Net}$,
where $C_\phi$ is either \code{noClock} or \code{clock($\tau$)}; the
latter represents the  time ($\tau$) since 
a $\phi$-state was visited, without having been followed by a $\psi$-state. 
 The rewrite rules are adjusted to 
update this new component as follows. 
 The new \code{tick} rule  updates 
 \code{clock(T1)} to \code{clock(T1$\;$+$\;$T)} and leaves 
 \code{noClock} unchanged. 
The rule \code{applyTransition} is split into two rules: 
\begin{maude}
crl [applyTransition] : clock(T) : M ... => NEW-TP : M' ...
 if NEW-TP := if $\mathit{STATE'}$ |= $\psi$ then noClock else clock(T) fi /\ ...
crl [applyTransition] : noClock : M ... => NEW-TP : M' ...
 if NEW-TP := if $\mathit{STATE'}$ |= $\phi$ and not $\mathit{STATE'}$ |= $\psi$ 
    then clock(0/1) else noClock fi /\ ...
\end{maude}
In the first rule, if a $\psi$-state is encountered, the
new ``$\phi$-clock''  is reset 
to \code{noClock}. In the second rule, this ``$\phi$-clock'' starts running
if the new state satisfies $\phi$ but not $\psi$. 
The query $\phi \rightsquigarrow_{\leq b} \psi $ can be answered
by   searching for a state where a $\phi$-state has not been followed
by a $\psi$-state before the deadline $b$:
\begin{maude}
search [1] ...  =>*  S : PHI' $\parallel $ clock(T) : ... such that T > $b$ .
\end{maude}

Reachability analysis cannot be used to analyze the other properties
supported by \romeo{} ($\textbf{Q}\, \phi\,\CTLU_J\,\psi$, and 
$\forall \CTLF_J\,\phi$ and its dual 
$\exists \CTLG_J \,\phi$).  While developing a full SMT-based \emph{timed} temporal logic model
checker is future 
work,  we can  combine Maude's explicit-state model checker and
SMT solving to solve these (and many other) queries. On the positive side, and
beyond \romeo, we can use full LTL, and also allow conditions on
clocks in state propositions.

The timed temporal operators can be defined on top of the
(untimed) LTL temporal operators in Maude (\code{<>}, \code{[]} and \code{U}) :
\begin{maude}
op <_>_ :  Interval Prop -> Formula .
op _U__ : Prop Interval Prop -> Formula .
op [_]_ : Interval Prop -> Formula .
vars PR1 PR2 : Prop .
eq < INTERVAL > PR1 = <> (PR1 /\ in-time INTERVAL) .
eq PR1 U INTERVAL PR2 = PR1 U (PR2 /\ in-time INTERVAL) .
eq [ INTERVAL ] PR1 = ~ (< INTERVAL > (~ PR1)) .
\end{maude}

For this fragment of  non-nested timed temporal logic formulas, it is possible
to model check universal and existential quantified formulas as
follows:

\begin{maude}
op A-model-check : State Formula -> Bool .  --- Universal queries
op E-model-check : State Formula -> Bool .  --- Existential queries
eq A-model-check(STATE, F) = modelCheck(STATE, F) == true .
eq E-model-check(STATE, F) = modelCheck(STATE , ~ F) =/= true .
\end{maude}

\section{Benchmarks}\label{sec:benchmarks}

We have compared the performance of our Maude-with-SMT analysis with that of
\romeo{} (version 3.8.6) on three case studies. We compare
the time  it takes for different rewrite theories to solve the synthesis
problem \EF{($p > n$)} (i.e., place $p$ holds more than $n$ tokens), for different
places $p$ and $0 \leq n \leq 2$, and to check
whether the net is $1$-safe.
The  models used in our experiments are: the
\texttt{producer-consumer}~\cite{Wang1998} system in~\Cref{fig:producers}, the
\texttt{scheduling}~\cite{DBLP:journals/jucs/TraonouezLR09} system
in~\Cref{fig:ex-scheduling}, and the \texttt{tutorial} system
in~\Cref{fig:tutorial} taken from the \romeo{} website. The model
\texttt{tutorial} was modified to produce two tokens in transition 
\texttt{startOver}, thus leading to infinite behaviors. 
 The details of each model can be found
in~\Cref{tab:models}.

We ran all the experiments on a Dell Precision Tower 3430 with a processor Intel
Xeon E-2136 6-cores @ 3.3GHz, 64 GiB memory, and Ubuntu 20.04. Each
experiment was executed using Maude in combination with two different SMT
solvers: Yices and Z3. We use a timeout of 10 minutes. 

\Cref{fig:bench} shows the execution times of \romeo{} and Maude in
log-scale, for the three case studies. (The data for each experiment
can be found in \Cref{app:data}). 
Each point in the figures represents the
time taken by \romeo{} and Maude to analyze the properties \EF{($p >
n$)} and \AG($\neg$\! 1-safe). The execution of  $\rtheorySymN{1}$ outperforms $\rtheorySymN{0}$
in some cases and 
the use of Yices2 shows better times when compared to 
Z3. For negative queries (e.g., \EF{($p > 2$)} is false
for \code{scheduling} and \code{producer-consumer}), as expected, we have
timeouts for $\rtheorySymN{0}$ and $\rtheorySymN{1}$. In those cases, 
$\rtheorySymNF{1}$ completes the analysis before the timeout. 
Currently, Maude-SE supports existential quantified queries
only with Z3 and $\rtheorySymNF{1}$ can be only executed
with that SMT solver. In the near future, Maude-SE will integrate
the support for quantifiers in Yices2 and we expect a better performance for 
$\rtheorySymNF{1}$. We finally note that, in some reachability queries, 
Maude-SE outperforms \romeo{}.
More interestingly, our approach terminates in cases where \romeo{} does not. Our
results are proven valid when injecting them in the model and running \romeo{} with
these additional constraints. This phenomenon happens when the search order leads
\romeo{} in the exploration of an infinite branch with an unbounded marking.

\begin{table}
  \centering
  \begin{tabular}{lrrrr}
\toprule
\textbf{model} & \textbf{parameters} & \textbf{places} & \textbf{transitions} & \textbf{arcs} \\
\midrule
producer_consumer & 1 & 5 & 4 & 10 \\
scheduling & 3 & 6 & 9 & 15 \\
tutorial & 2 & 6 & 5 & 12 \\
\bottomrule
\end{tabular}
   \caption{Description of the models used in the benchmarks
  (\Cref{fig:bench}).\label{tab:models}}
\end{table}

\begin{figure}
    \centering
    \includegraphics[width=0.5\textwidth]{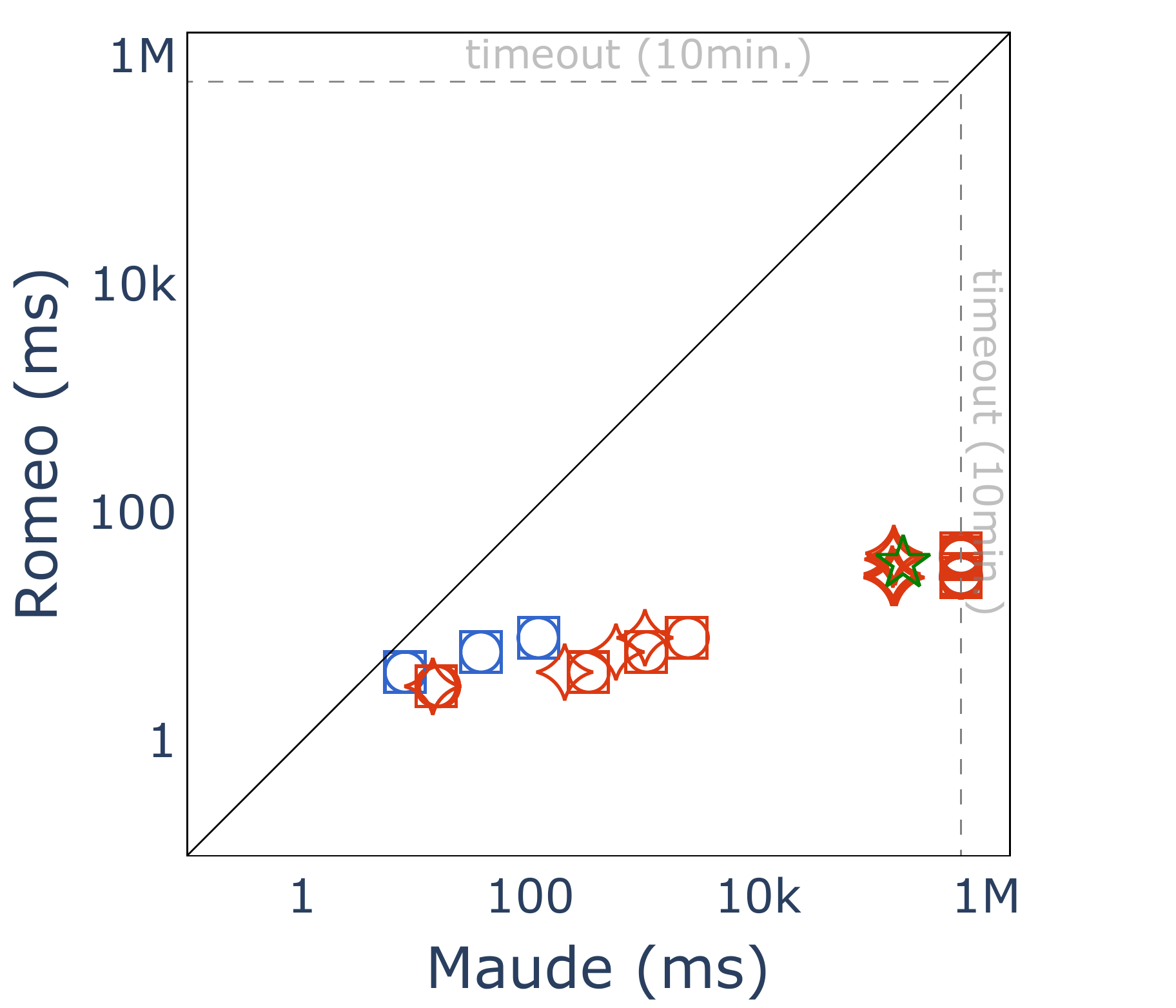}
    \caption{Case study in \cite{DBLP:journals/jucs/TraonouezLR09}. \label{fig:ex-scheduling}}
\end{figure}

\begin{figure}
    \centering
    \includegraphics[width=.6\textwidth]{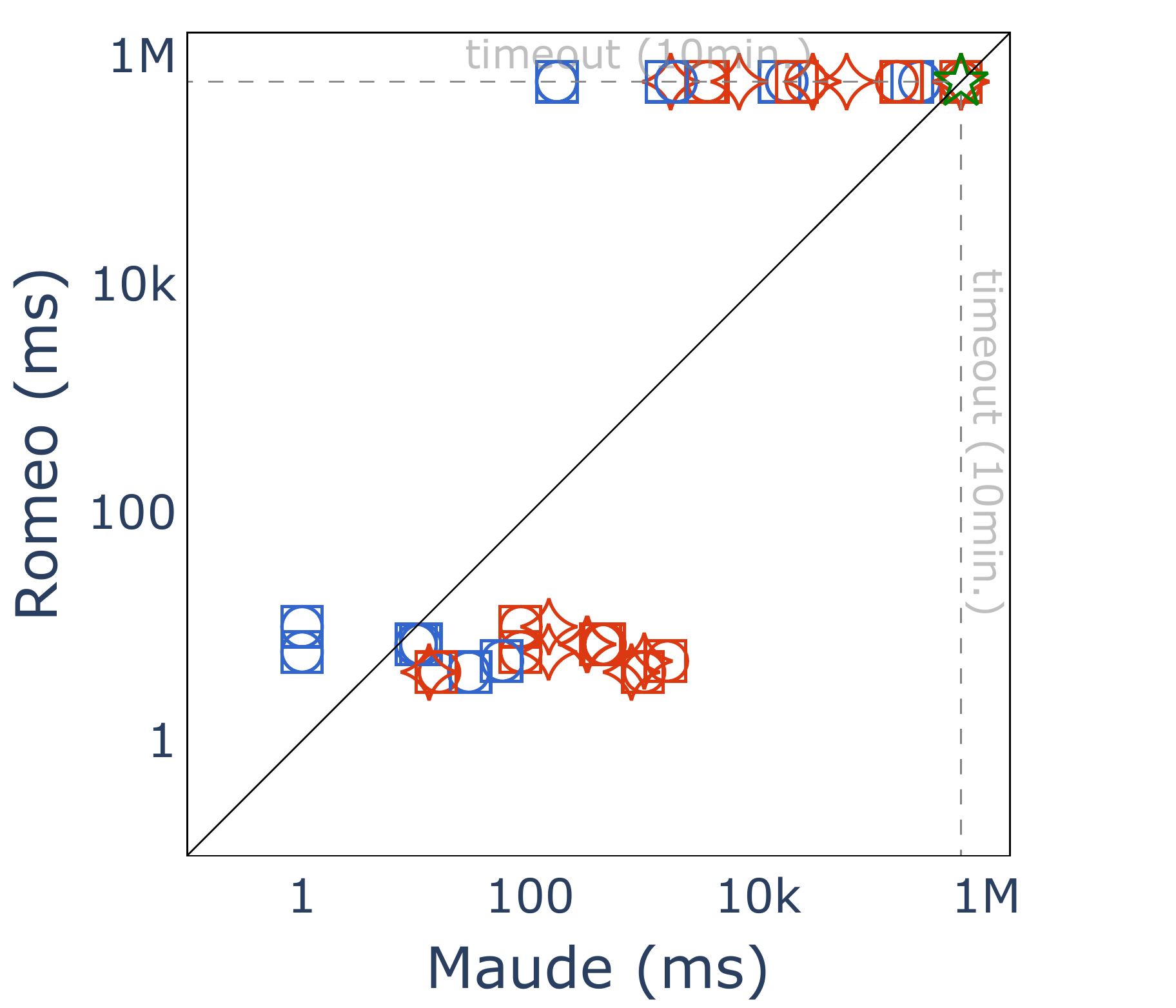}
    \caption{Case study tutorial. \label{fig:tutorial}}
\end{figure}
\begin{figure}[!htb]
  \centering
  \begin{subfigure}[b]{0.45\textwidth}
    \centering
    \includegraphics[width=\textwidth]{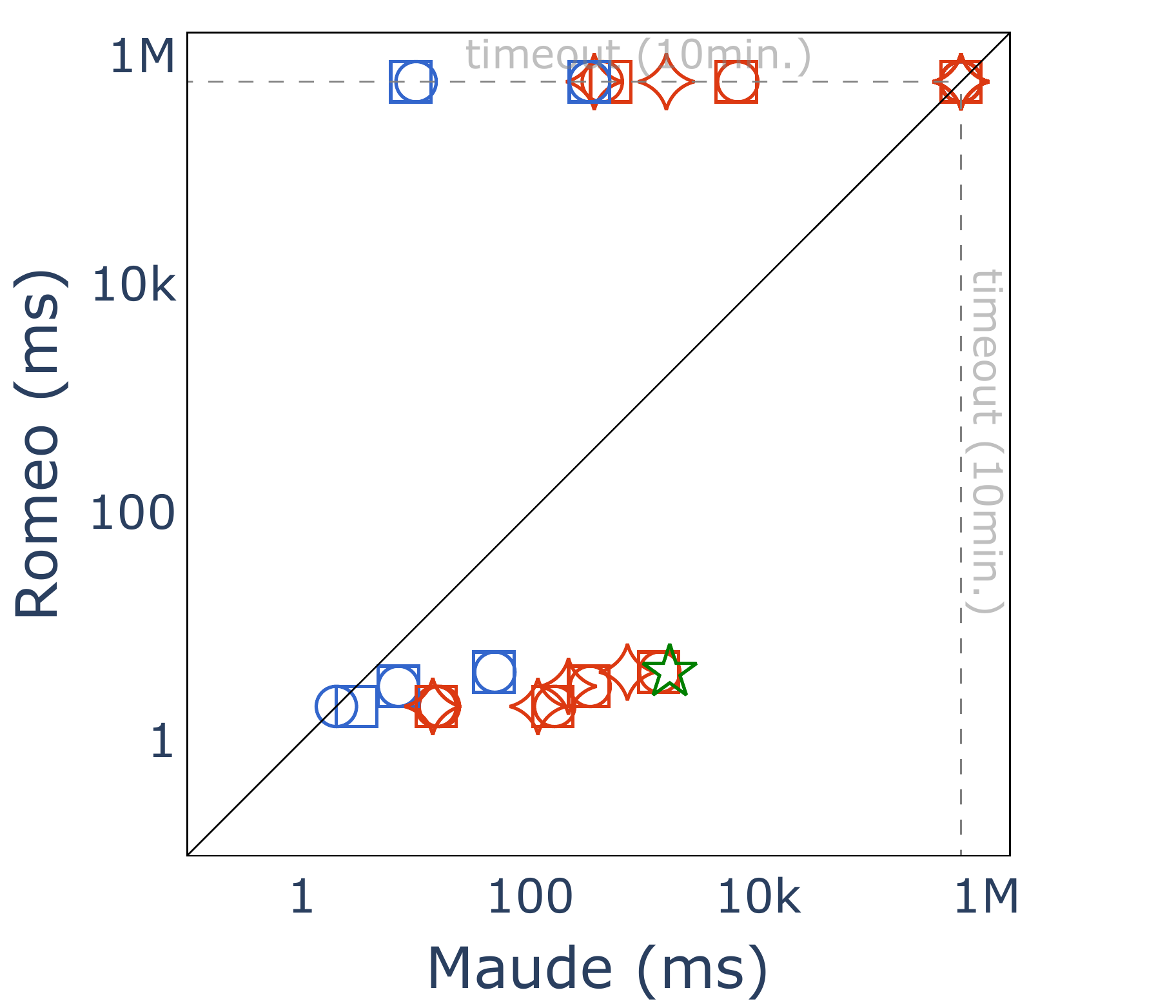}
    \caption{\texttt{producer-consumer}}
  \end{subfigure}
  \begin{subfigure}[b]{0.45\textwidth}
    \centering
    \includegraphics[width=\textwidth]{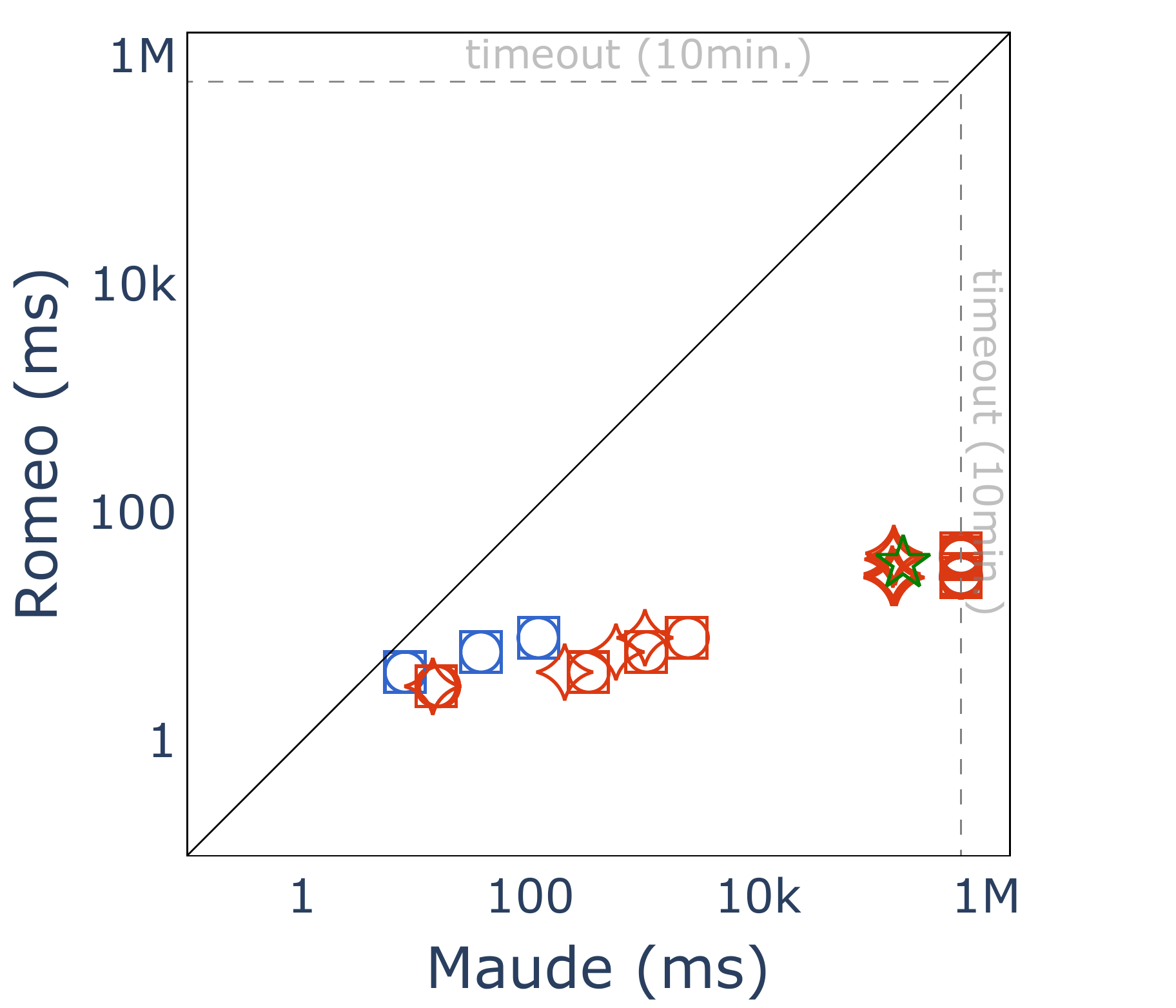}
    \caption{\texttt{scheduling}}
  \end{subfigure}
  \begin{subfigure}[b]{0.45\textwidth}
    \centering
    \includegraphics[width=\textwidth]{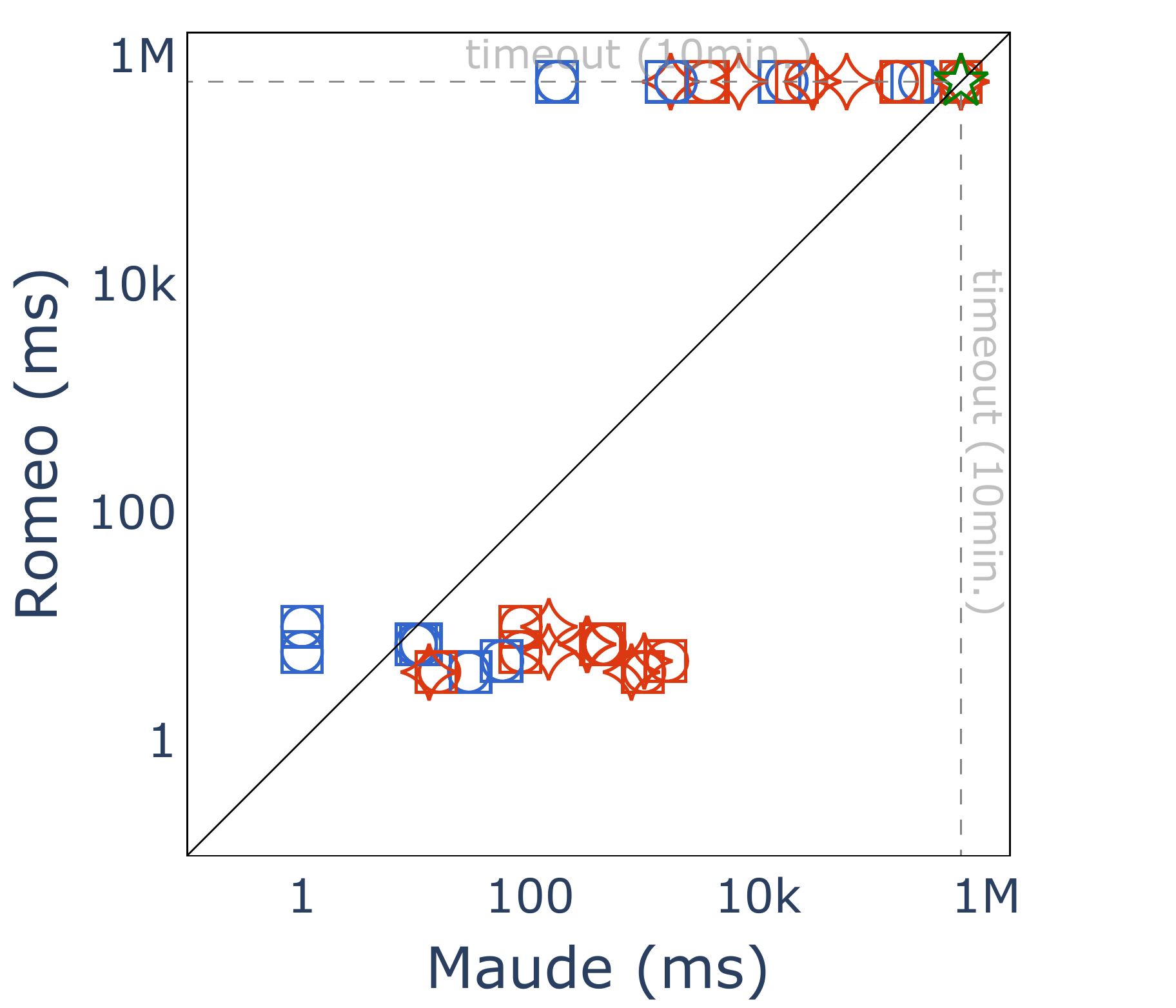}
    \caption{\texttt{tutorial}}
  \end{subfigure}
  \caption{Execution times for \romeo{} and Maude,  in log-scale.  Theory
  $\rtheorySymN{0}$ using Z3 (\symOneZ), theory $\rtheorySymN{0}$ using Yices
  (\symOneYices), theory $\rtheorySymN{1}$ using Z3 (\symTwoZ), theory
  $\rtheorySymN{1}$ using Yices (\symTwoYices), and $\rtheorySymNF{1}$  using Z3
  (\foldingTheory). Point \safeTheory{} is the property 1-safe.\label{fig:bench}}
\end{figure}

\section{Related Work}   \label{sec:related}

\paragraph{Tool support for parametric time Petri nets.}
We are not aware of any other tool for analyzing parametric time(d)
Petri nets than \romeo{}~\cite{romeo}. 

\paragraph{Petri nets in rewriting logic.} Formalizing
Petri nets algebraically~\cite{meseguer-montanari}  was one of the
inspirations behind rewriting  logic. Different kinds of Petri nets
are given a rewriting logic semantics in~\cite{petri-nets-in-maude}, 
and in~\cite{OlvMesTCS} for timed nets. In contrast to our paper, these papers 
focus on the semantics of such nets, and do not consider execution and
analysis;  nor do they consider inhibitor arcs or parameters.
Capra~\cite{capra22,capra22b}, Padberg and Schultz~\cite{padberg},  and Barbosa et
al.~\cite{barbosa} use  
Maude to formalize   dynamically 
reconfigurable Petri nets (with inhibitor arcs) and I/O
Petri nets.  In
contrast to our work, these papers target untimed and non-parametric nets, and
do not focus on formal analysis, but only show examples
of standard (explicit-state) \lstinline{search}  and LTL model
checking.  

\paragraph{Symbolic methods for real-time systems in Maude.}
We  develop a symbolic rewrite
semantics and analysis for parametric time automata (PTA)
in~\cite{ftscs22}.  The differences with the current paper include:
PTAs are very simple structures 
compared to PITPNs (with inhibitor arcs, no bounds on the number of
tokens in a state), so that the semantics of PITPNs is  more
sophisticated than the one for PTAs, which does not use
``structured'' states, equations, or user-defined functions; defining
a new rewrite theory for 
each PTA in~\cite{ftscs22} compared to having a single rewrite theory
for all nets in this work; obtaining desired symbolic reachability
properties using ``standard'' folding of symbolic states for PTAs compared to having to
develop a new folding mechanism  for PITPNs; analysis in \cite{ftscs22} do not include
model checking temporal logic formulas; and so on.

In addition, 
a variety of real-time systems 
have been formally analyzed using rewriting with SMT,
including 
PLC ST programs~\cite{lee2022bounded},
virtually synchronous cyber-physical systems~\cite{lee2022extension,hsaadl-sttt,lee2021hybrid},
and
soft agents~\cite{nigam2022automating}.
These papers differ from our work
in that 
they use 
guarded terms~\cite{bae2017guarded,bae2019symbolic}
for state-space reduction instead of folding, 
and 
do not consider parameter synthesis problems.

\section{Concluding Remarks}  \label{sec:concl}

We have provided a ``concrete'' rewriting logic semantics for PITPNs, and proved that this semantics is bisimilar to the semantics
of such nets in~\cite{paris-paper}. However, this model is
non-executable; furthermore, explicit-state Maude analysis using
Real-Time Maude-style ``time sampling'' leads to unsound analysis for
dense-time systems such as PITPNs. We therefore \emph{systematically
  transformed} this model  into a ``symbolic'' rewrite model which  is
amenable to  sound and complete symbolic  analysis using Maude
combined with SMT solving. 

We have shown how  almost all formal analysis and parameter synthesis
supported by the PITPN tool \romeo{} can be performed using
Maude-with-SMT. In addition, we have shown how Maude-with-SMT
can provide  additional capabilities for PITPNs, including synthesizing
 initial markings (and not just firing bounds) from \emph{parametric}
initial markings so that desired properties are satisfied,  full LTL
model checking, and analysis with user-defined execution
strategies.  We have developed  a new
``folding''  method for symbolic states, so that symbolic reachability
analysis using Maude-with-SMT terminates whenever  the corresponding
\romeo{} analysis terminates.

We have compared the performance of \romeo{} and our Maude-with-SMT methods on
a number of benchmarks, which show that Maude combined with the SMT solver
Yices in many cases outperforms \romeo, whereas Maude combined with Z3 is
significantly slower.  We also experienced that \romeo{} sometimes did not find
(existing) solutions and the output of some executions included the message
``\emph{maybe}'', showing that \romeo\ was computing an approximation. As
mentioned in \Cref{sec:benchmarks}, this can be caused by the search
exploration mechanism implemented in \romeo. 
Maude's search commands use a
breadth-first strategy, thus guaranteeing completeness (if a given state is
reachable, it will be eventually found). Moreover, operations on constraints are
delegated to state-of-the-art SMT solvers.
We also point out
that Maude's specifications are very close to their corresponding 
mathematical definitions. Hence, it is easier to check the correctness
of the implementation  and, together with Maude's meta-programming
features, it is easy to develop, test and evaluate different analysis
algorithms.

This paper has not only provided new features   for PITPNs. It has
also shown that even a 
model like our Real-Time Maude-inspired PITPN interpreter---with
functions, equations, and unbounded markings---can easily be turned
into a  symbolic rewrite theory for which Maude-with-SMT provides very 
useful sound and complete analyses even for dense-time systems.

In future work we should: implement the needed Maude-SE's bindings for quantifiers in Yices2,
thus improving the performance of analysis with $\rtheorySymNF{1}$;
extend Maude's LTL model checker to a full SMT-based (with folding) timed
LTL and CTL model checker, thus covering all the analysis provided  by
\romeo{}; develop a richer timed strategy language for controlling 
the executions of PITPNs; and explore theory transformations
for the sound and complete symbolic analysis of Real-Time Maude
specifications.

\appendix
\section{Data for the benchmarks}\label{app:data}
\begin{table}[!htb]
  \centering
  \rowcolors{4}{gray!25}{white}
  \resizebox{\columnwidth}{!}{\begin{tabular}{l|l|c|cc||cc||c}
  \toprule
  \multirow{3}{*}{\textbf{Model}} & \multirow{3}{*}{\textbf{Place reached}} & \multirow{3}{*}{\textbf{\romeo{} (ms)}} & \multicolumn{5}{c}{\textbf{Maude (ms)}}                                                                                                           \\
                                  &                                         &                                        & \multicolumn{2}{c||}{$\rtheorySymN{0}$} & \multicolumn{2}{c||}{$\rtheorySymN{1}$} & $\rtheorySymNF{1}$                                           \\
                                  &                                         &                                        & \textbf{Yices}                          & \textbf{Z3}                             & \textbf{Yices}                   & \textbf{Z3} & \textbf{Z3} \\
  \midrule
  tutorial                        & start                                   & 4.0                                    & \color{OliveGreen} \bfseries 0.0        & 16.0                                    & \color{OliveGreen} \bfseries 0.0 & 15.0        & 13.0        \\
  tutorial                        & childStart                              & 6.0                                    & \color{OliveGreen} \bfseries 1.0        & 83.0                                    & \color{OliveGreen} \bfseries 1.0 & 82.0        & 145.0       \\
  tutorial                        & fatherCont                              & 10.0                                   & \color{OliveGreen} \bfseries 1.0        & 83.0                                    & \color{OliveGreen} \bfseries 1.0 & 82.0        & 146.0       \\
  tutorial                        & childDone                               & \color{OliveGreen} \bfseries 7.0       & 10.0                                    & 424.0                                   & 10.0                             & 414.0       & 320.0       \\
  tutorial                        & fatherDone                              & \color{OliveGreen} \bfseries 7.0       & 11.0                                    & 453.0                                   & 11.0                             & 445.0       & 313.0       \\
  tutorial                        & joined                                  & \color{OliveGreen} \bfseries 4.0       & 29.0                                    & 1005.0                                  & 30.0                             & 977.0       & 773.0       \\
  producer_consumer               & itemReady                               & \color{OliveGreen} \bfseries 2.0       & \color{OliveGreen} \bfseries 2.0        & 164.0                                   & 3.0                              & 158.0       & 117.0       \\
  producer_consumer               & buffer                                  & \color{OliveGreen} \bfseries 3.0       & 7.0                                     & 336.0                                   & 7.0                              & 327.0       & 217.0       \\
  producer_consumer               & itemReceived                            & \color{BrickRed} \bfseries TO          & 10.0                                    & 429.0                                   & \color{OliveGreen} \bfseries 9.0 & 509.0       & 365.0       \\
  producer_consumer               & readyConsumer                           & 2.0                                    & \color{OliveGreen} \bfseries 0.0        & 16.0                                    & \color{OliveGreen} \bfseries 0.0 & 15.0        & 14.0        \\
  producer_consumer               & readyProducer                           & 2.0                                    & \color{OliveGreen} \bfseries 0.0        & 15.0                                    & \color{OliveGreen} \bfseries 0.0 & 15.0        & 14.0        \\
  scheduling                      & ready1                                  & 3.0                                    & \color{OliveGreen} \bfseries 0.0        & 16.0                                    & \color{OliveGreen} \bfseries 0.0 & 15.0        & 14.0        \\
  scheduling                      & ready2                                  & 3.0                                    & \color{OliveGreen} \bfseries 0.0        & 16.0                                    & \color{OliveGreen} \bfseries 0.0 & 15.0        & 14.0        \\
  scheduling                      & ready3                                  & 3.0                                    & \color{OliveGreen} \bfseries 0.0        & 15.0                                    & \color{OliveGreen} \bfseries 0.0 & 15.0        & 14.0        \\
  scheduling                      & ending1                                 & \color{OliveGreen} \bfseries 4.0       & 8.0                                     & 327.0                                   & 8.0                              & 324.0       & 201.0       \\
  scheduling                      & ending2                                 & \color{OliveGreen} \bfseries 6.0       & 37.0                                    & 1054.0                                  & 37.0                             & 1039.0      & 566.0       \\
  scheduling                      & ending3                                 & \color{OliveGreen} \bfseries 8.0       & 118.0                                   & 2422.0                                  & 119.0                            & 2364.0      & 1015.0      \\
  \bottomrule
\end{tabular}
 }
  \caption{times for \EF{($p > 0$)}}
\end{table}

\begin{table}[!htb]
  \centering
  \rowcolors{4}{gray!25}{white}
  \resizebox{\columnwidth}{!}{\begin{tabular}{l|l|c|cc||cc||c}
  \toprule
  \multirow{3}{*}{\textbf{Model}} & \multirow{3}{*}{\textbf{Place reached}} & \multirow{3}{*}{\textbf{\romeo{} (ms)}} &                                                                                                                    \multicolumn{5}{c}{\textbf{Maude (ms)}}                                                                                                                                                        \\
                                  &                                         &                                        & \multicolumn{2}{c||}{$\rtheorySymN{0}$} & \multicolumn{2}{c||}{$\rtheorySymN{1}$} & $\rtheorySymNF{1}$                                                                                        \\
                                  &                                         &                                        & \textbf{Yices}                          & \textbf{Z3}                             & \textbf{Yices}                     & \textbf{Z3}                   & \textbf{Z3}                          \\
  \midrule
  tutorial                        & start                                   & \color{OliveGreen} \bfseries 5.0       & 57.0                                    & 1594.0                                  & 56.0                               & 1541.0                        & 1001.0                               \\
  tutorial                        & childStart                              & \color{BrickRed} \bfseries TO          & \color{OliveGreen} \bfseries 172.0      & 3584.0                                  & \color{OliveGreen} \bfseries 172.0 & 3459.0                        & 1705.0                               \\
  tutorial                        & fatherCont                              & \color{BrickRed} \bfseries TO          & \color{OliveGreen} \bfseries 171.0      & 3604.0                                  & 172.0                              & 3608.0                        & 1710.0                               \\
  tutorial                        & childDone                               & \color{BrickRed} \bfseries TO          & \color{BrickRed} \bfseries TO           & \color{BrickRed} \bfseries TO           & \color{BrickRed} \bfseries TO      & \color{BrickRed} \bfseries TO & \color{BrickRed} \bfseries TO        \\
  tutorial                        & fatherDone                              & \color{BrickRed} \bfseries TO          & 262316.0                                & \color{BrickRed} \bfseries TO           & 224865.0                           & \color{BrickRed} \bfseries TO & \color{OliveGreen} \bfseries 59594.0 \\
  tutorial                        & joined                                  & \color{BrickRed} \bfseries TO          & \color{BrickRed} \bfseries TO           & \color{BrickRed} \bfseries TO           & \color{BrickRed} \bfseries TO      & \color{BrickRed} \bfseries TO & \color{BrickRed} \bfseries TO        \\
  producer_consumer               & itemReady                               & \color{BrickRed} \bfseries TO          & \color{BrickRed} \bfseries TO           & \color{BrickRed} \bfseries TO           & \color{BrickRed} \bfseries TO      & \color{BrickRed} \bfseries TO & \color{BrickRed} \bfseries TO        \\
  producer_consumer               & buffer                                  & \color{OliveGreen} \bfseries 4.0       & 49.0                                    & 1368.0                                  & 48.0                               & 1338.0                        & 713.0                                \\
  producer_consumer               & itemReceived                            & \color{BrickRed} \bfseries TO          & \color{BrickRed} \bfseries TO           & \color{BrickRed} \bfseries TO           & \color{BrickRed} \bfseries TO      & \color{BrickRed} \bfseries TO & \color{BrickRed} \bfseries TO        \\
  producer_consumer               & readyConsumer                           & \color{BrickRed} \bfseries TO          & \color{BrickRed} \bfseries TO           & \color{BrickRed} \bfseries TO           & \color{BrickRed} \bfseries TO      & \color{BrickRed} \bfseries TO & \color{BrickRed} \bfseries TO        \\
  producer_consumer               & readyProducer                           & \color{BrickRed} \bfseries TO          & \color{BrickRed} \bfseries TO           & \color{BrickRed} \bfseries TO           & \color{BrickRed} \bfseries TO      & \color{BrickRed} \bfseries TO & \color{BrickRed} \bfseries TO        \\
  scheduling                      & ready1                                  & \color{OliveGreen} \bfseries 27.0      & \color{BrickRed} \bfseries TO           & \color{BrickRed} \bfseries TO           & \color{BrickRed} \bfseries TO      & \color{BrickRed} \bfseries TO & 154670.0                             \\
  scheduling                      & ready2                                  & \color{OliveGreen} \bfseries 27.0      & \color{BrickRed} \bfseries TO           & \color{BrickRed} \bfseries TO           & \color{BrickRed} \bfseries TO      & \color{BrickRed} \bfseries TO & 161968.0                             \\
  scheduling                      & ready3                                  & \color{OliveGreen} \bfseries 27.0      & \color{BrickRed} \bfseries TO           & \color{BrickRed} \bfseries TO           & \color{BrickRed} \bfseries TO      & \color{BrickRed} \bfseries TO & 150771.0                             \\
  scheduling                      & ending1                                 & \color{OliveGreen} \bfseries 41.0      & \color{BrickRed} \bfseries TO           & \color{BrickRed} \bfseries TO           & \color{BrickRed} \bfseries TO      & \color{BrickRed} \bfseries TO & 158852.0                             \\
  scheduling                      & ending2                                 & \color{OliveGreen} \bfseries 27.0      & \color{BrickRed} \bfseries TO           & \color{BrickRed} \bfseries TO           & \color{BrickRed} \bfseries TO      & \color{BrickRed} \bfseries TO & 161629.0                             \\
  scheduling                      & ending3                                 & \color{OliveGreen} \bfseries 27.0      & \color{BrickRed} \bfseries TO           & \color{BrickRed} \bfseries TO           & \color{BrickRed} \bfseries TO      & \color{BrickRed} \bfseries TO & 157329.0                             \\
  \bottomrule
\end{tabular}
 }
  \caption{times for \EF{($p > 1$)}}
\end{table}

\begin{table}[!htb]
  \centering
  \rowcolors{4}{gray!25}{white}
  \resizebox{\columnwidth}{!}{\begin{tabular}{l|l|c|cc||cc||c}
  \toprule
  \multirow{3}{*}{\textbf{Model}} & \multirow{3}{*}{\textbf{Place reached}} & \multirow{3}{*}{\textbf{\romeo{} (ms)}} & \multicolumn{5}{c}{\textbf{Maude (ms)}}                                                                                                                                                   \\
                                  &                                         &                                        & \multicolumn{2}{c||}{$\rtheorySymN{0}$} & \multicolumn{2}{c||}{$\rtheorySymN{1}$} & $\rtheorySymNF{1}$                                                                                   \\
                                  &                                         &                                        & \textbf{Yices}                          & \textbf{Z3}                             & \textbf{Yices}                       & \textbf{Z3}                   & \textbf{Z3}                   \\
  \midrule
  tutorial                        & start                                   & \color{BrickRed} \bfseries TO          & 1904.0                                  & 26445.0                                 & \color{OliveGreen} \bfseries 1572.0  & 21787.0                       & 6799.0                        \\
  tutorial                        & childStart                              & \color{BrickRed} \bfseries TO          & 17760.0                                 & 164884.0                                & \color{OliveGreen} \bfseries 15447.0 & 180210.0                      & 30043.0                       \\
  tutorial                        & fatherCont                              & \color{BrickRed} \bfseries TO          & 17813.0                                 & 164248.0                                & \color{OliveGreen} \bfseries 15432.0 & 181712.0                      & 30186.0                       \\
  tutorial                        & childDone                               & \color{BrickRed} \bfseries TO          & \color{BrickRed} \bfseries TO           & \color{BrickRed} \bfseries TO           & \color{BrickRed} \bfseries TO        & \color{BrickRed} \bfseries TO & \color{BrickRed} \bfseries TO \\
  tutorial                        & fatherDone                              & \color{BrickRed} \bfseries TO          & \color{BrickRed} \bfseries TO           & \color{BrickRed} \bfseries TO           & \color{BrickRed} \bfseries TO        & \color{BrickRed} \bfseries TO & \color{BrickRed} \bfseries TO \\
  tutorial                        & joined                                  & \color{BrickRed} \bfseries TO          & \color{BrickRed} \bfseries TO           & \color{BrickRed} \bfseries TO           & \color{BrickRed} \bfseries TO        & \color{BrickRed} \bfseries TO & \color{BrickRed} \bfseries TO \\
  producer_consumer               & itemReady                               & \color{BrickRed} \bfseries TO          & \color{BrickRed} \bfseries TO           & \color{BrickRed} \bfseries TO           & \color{BrickRed} \bfseries TO        & \color{BrickRed} \bfseries TO & \color{BrickRed} \bfseries TO \\
  producer_consumer               & buffer                                  & \color{BrickRed} \bfseries TO          & 333.0                                   & 6618.0                                  & \color{OliveGreen} \bfseries 331.0   & 6419.0                        & 1564.0                        \\
  producer_consumer               & itemReceived                            & \color{BrickRed} \bfseries TO          & \color{BrickRed} \bfseries TO           & \color{BrickRed} \bfseries TO           & \color{BrickRed} \bfseries TO        & \color{BrickRed} \bfseries TO & \color{BrickRed} \bfseries TO \\
  producer_consumer               & readyConsumer                           & \color{BrickRed} \bfseries TO          & \color{BrickRed} \bfseries TO           & \color{BrickRed} \bfseries TO           & \color{BrickRed} \bfseries TO        & \color{BrickRed} \bfseries TO & \color{BrickRed} \bfseries TO \\
  producer_consumer               & readyProducer                           & \color{BrickRed} \bfseries TO          & \color{BrickRed} \bfseries TO           & \color{BrickRed} \bfseries TO           & \color{BrickRed} \bfseries TO        & \color{BrickRed} \bfseries TO & \color{BrickRed} \bfseries TO \\
  scheduling                      & ready1                                  & \color{OliveGreen} \bfseries 44.0      & \color{BrickRed} \bfseries TO           & \color{BrickRed} \bfseries TO           & \color{BrickRed} \bfseries TO        & \color{BrickRed} \bfseries TO & 154580.0                      \\
  scheduling                      & ready2                                  & \color{OliveGreen} \bfseries 27.0      & \color{BrickRed} \bfseries TO           & \color{BrickRed} \bfseries TO           & \color{BrickRed} \bfseries TO        & \color{BrickRed} \bfseries TO & 157433.0                      \\
  scheduling                      & ready3                                  & \color{OliveGreen} \bfseries 27.0      & \color{BrickRed} \bfseries TO           & \color{BrickRed} \bfseries TO           & \color{BrickRed} \bfseries TO        & \color{BrickRed} \bfseries TO & 157285.0                      \\
  scheduling                      & ending1                                 & \color{OliveGreen} \bfseries 29.0      & \color{BrickRed} \bfseries TO           & \color{BrickRed} \bfseries TO           & \color{BrickRed} \bfseries TO        & \color{BrickRed} \bfseries TO & 149607.0                      \\
  scheduling                      & ending2                                 & \color{OliveGreen} \bfseries 27.0      & \color{BrickRed} \bfseries TO           & \color{BrickRed} \bfseries TO           & \color{BrickRed} \bfseries TO        & \color{BrickRed} \bfseries TO & 148758.0                      \\
  scheduling                      & ending3                                 & \color{OliveGreen} \bfseries 39.0      & \color{BrickRed} \bfseries TO           & \color{BrickRed} \bfseries TO           & \color{BrickRed} \bfseries TO        & \color{BrickRed} \bfseries TO & 151563.0                      \\
  \bottomrule
\end{tabular}
 }
  \caption{times for \EF{($p > 2$)}}
\end{table}

\begin{table}[!htb]
  \centering
  \rowcolors{2}{gray!25}{white}
  \begin{tabular}{lrr}
  \toprule
  \textbf{Model}    & \textbf{\romeo{} (ms)}             &
                                                           \textbf{Maude (ms)}            \\
  \midrule
  tutorial          & \color{BrickRed} \bfseries TO     & \color{BrickRed} \bfseries TO \\
  producer_consumer & \color{OliveGreen} \bfseries 4.0  & 1676.0                        \\
  scheduling        & \color{OliveGreen} \bfseries 36.0 & 186624.0                      \\
  \bottomrule
\end{tabular}
   \caption{1-safe time verification}
\end{table}

\section{Proofs of the Results}

\subsection{Proof of Theorem \ref{thm:ground}}

\begin{proof}
  (i) By definition $a_0 = (\marking_0,\parIntervalStatic) \approx (\encBase{\marking_0}
  \;\code{:}\;\mathtt{initClocks}(\encBase{\PN})\;\code{:}\;\encBase{\PN})$,
  since all clocks are \texttt{0} in $\mathtt{initClocks}(...)$, so
  that these clocks satisfy all the constraints in
  Definition~\ref{def:bisim-relation} since  $I=J$ in the initial state. 
  (ii) Follows from the following two lemmas. 
\end{proof}

\begin{lemma}
If $\left(\marking,\interval\right) \to
\left(\marking',\interval'\right)$
and $\left(\marking,\interval\right) \approx
(\encBase{\marking}\;\code{:}\;\mathit{clocks}\;\code{:}\;\encBase{\PN})$
then there is a $\mathit{clocks'}$ such that
$(\encBase{\marking}
\;\code{:}\;\mathit{clocks}\;\code{:}\;\encBase{\PN}) 
\mapsto (\encBase{\marking'}\;\code{:}\;\mathit{clocks'}
\;\code{:}\;\encBase{\PN})$
and $\left(\marking',\interval'\right) \approx
(\encBase{\marking'}\;\code{:}\;\mathit{clocks'}\;\code{:}\;\encBase{\PN})$.
\end{lemma}


\begin{proof}
Since $\left(\marking,\interval\right) \to \left(\marking',\interval'\right)$, we 
have that there exists an intermediate pair $(\marking,\interval'') \in (\Transition\cup\grandrplus)$ such that
$\left(\marking,\interval\right) \fleche{\delta} \left(\marking,\interval''\right)$ and 
$\left(\marking,\interval''\right) \fleche{\transition_f} \left(\marking',\interval'\right)$.

\noindent
For the first step ($\fleche{\delta}$), since $\left(\marking,\interval\right) \fleche{\delta} \left(\marking,\interval''\right)$,
there exists a $\delta$ such that $\forall t \in \Transition$, either $\interval''(\transition) = 
\interval(\transition)$ or $\leftEP{\interval''(\transition)} = \max(0,\leftEP{\interval(\transition)} - \delta)$
and $\rightEP{\interval''(\transition)} = \rightEP{\interval(\transition)} - \delta$. In both cases
we have that $\forall \transition \in \Transition, \rightEP{\interval''(\transition)}\geq 0$.
Now, letting \code{T} $= \delta$, it must be the case that 
\code{T <= mte($\encBase{\marking}$, $clocks$, $\encBase{\PN}$)}. This is 
because \code{mte($\encBase{\marking}$, $clocks$, $\encBase{\PN}$)} is 
defined to be equal to the minimum difference between  
$\rightEP{\parIntervalStatic(\transition)}$ and the clock value of $t$ out of all
$\transition \in \Transition$. That is, it is the maximum time that can elapse before an 
enabled transition reaches the right endpoint of its interval. In other words, an upper limit
for $\delta$. Hence, the \code{tick}-rule can be applied to 
$(\encBase{\marking}\code{ :
}\mathit{clocks}\code{ : }\encBase{\PN})$ 
with all enabled clocks having their time advanced by $\delta$.

\noindent
For the second step ($\fleche{\transition_f}$), since $\left(\marking,\interval''\right) 
\fleche{\transition_f} \left(\marking',\interval'\right)$, the transition $\transition_f$
is active and $\leftEP{\interval''(\transition_f)} = 0$. Since
$\leftEP{\interval(\transition_f)} = 0$, 
the clock of transition $\transition_f$ must be in the interval $[\leftEP{\parIntervalStatic(\transition_f)},
\rightEP{\parIntervalStatic(\transition_f)}]$ by definition of $\mathit{clocks}$ for $(\marking,\interval'')$.
This is precisely the condition for applying the \code{applyTransition}-rule to the 
resulting state of the previous \code{tick}-rule application.
\end{proof}

\begin{lemma}
If
$(\encBase{\marking}\;\code{:}\;\mathit{clocks}\;\code{:}\;\encBase{\PN})
\mapsto b$  and \newline  $\left(\marking, \interval\right) \approx 
(\encBase{\marking}\;\code{:}\;\mathit{clocks}\;\code{:}\;\encBase{\PN})$, 
then there exists a state $\left(\marking', \interval'\right) \in \mathcal{A}$ 
such that $\left(\marking, \interval\right) \to \left(\marking', \interval'\right)$ and 
$\left(\marking', \interval'\right) \approx b$. \end{lemma}

\begin{proof}
Since $(\encBase{\marking}\;\code{:}\;\mathit{clocks}\;\code{:}\;\encBase{\PN}) \mapsto b$, 
we have that there exists an intermediate state $(\code{M : CLOCKS : NET}) \in T_{\Sigma , \texttt{State}}$
such that $(\encBase{\marking}\;\code{:}\;\mathit{clocks}\;\code{:}\;\encBase{\PN}) 
\stackrel{\texttt{tick}}{\longrightarrow} (\code{M : CLOCKS : NET})$ and 
$(\code{M : CLOCKS : NET}) \stackrel{\texttt{applyTransition}}{\longrightarrow} b$.

\noindent
For the first step $\left(\stackrel{\texttt{tick}}{\longrightarrow}\right)$, since 
$(\encBase{\marking}\;\code{:}\;\mathit{clocks}\;\code{:}\;\encBase{\PN}) 
\stackrel{\texttt{tick}}{\longrightarrow} (\code{M : CLOCKS : NET})$, there is a 
\code{T <= mte($\encBase{\marking}$, $clocks$, $\encBase{\PN}$)}. 
Now, as in the previous lemma, since the \code{mte} is an upper limit for $\delta$, 
we have that there exists a time transition
$\left(\marking,\interval\right) \fleche{\delta} \left(\marking,\interval''\right)$
with $\delta$ equal to the \code{T} used in the above \code{tick}-rule application so that
$(\code{M : CLOCKS : NET}) = (\encBase{\marking}\;\code{:}\;\mathit{clocks}\;\code{:}\;\encBase{\PN})$. 

\noindent
For the second step $\left(\stackrel{\texttt{applyTransition}}{\longrightarrow}\right)$, since
$(\encBase{\marking}\;\code{:}\;\mathit{clocks}\;\code{:}\;\encBase{\PN})
\stackrel{\texttt{applyTransition}}{\longrightarrow} b$, there must be a transition $\transition_f$ which 
is active and whose clock is in the interval $[\leftEP{\parIntervalStatic(\transition_f)},
\rightEP{\parIntervalStatic(\transition_f)}]$. By definition of $\mathit{clocks}$ on $(\marking,\interval'')$,
This is precisely the condition for the discrete step $\transition_f$ to $(\marking,\interval'')$.
Hence, $\left(\marking,\interval''\right) \fleche{\transition_f} \left(\marking',\interval'\right)$ and
$(\encBase{\marking}\;\code{:}\;\mathit{clocks}\;\code{:}\;\encBase{\PN})
\stackrel{\texttt{applyTransition}}{\longrightarrow} 
(\encBase{\marking'}\;\code{:}\;\mathit{clocks}\;\code{:}\;\encBase{\PN})$.
\end{proof}

\subsection{Proof of Theorems \ref{th:r0r1} and \ref{th:r0r2}}

\noindent\textbf{Theorem \ref{th:r0r1}}. 
\begin{proof}
For the ($\Leftarrow$) side, it suffices to follow in $\rtheorySem$ the
    same execution strategy as in $\rtheorySemN{1}$. 
    For ($\Rightarrow$), it suffices to perform the following
    (reachability-preserving) change in the $\rtheorySem$ trace: the
    application of two consecutive \code{tick} rules with $T=t_1$ and $T=t_2$
    are replaced by a single application of \code{tick} with $T=t_1 +
    t_2$. 
    This is enough to show that the same trace can be obtained in $\rtheorySemN{1}$. 
\end{proof}

\noindent\textbf{Theorem \ref{th:r0r2}}. 
\begin{proof}
    From \Cref{th:r0r1}, we know that the \code{tickOk}/\code{tickNotOk}
    strategy can be followed in $\rtheorySemN{0}$ to produce
    an equivalent trace. Using that trace, the result follows trivially
    by noticing that applications of 
    \code{tick} in $\rtheorySemN{0}$ with \code{T}$=\delta$
    ($t \dtransitionD{\delta}_{\rtheorySem} t'$)
    match applications of \code{tick} in 
     $\rtheorySemN{2}$ with the same 
    instance of \code{T}, thus advancing the global clock in 
    exactly $\delta$ time-units.
\end{proof}

\subsection{Proof of Theorems \ref{th:sym-correct} and \ref{th:folding}}

\noindent\textbf{Theorem \ref{th:sym-correct}}

\begin{proof}
This result is a direct consequence of 
soundness and completeness of rewriting modulo SMT \cite{rocha-rewsmtjlamp-2017}.
More precisely, from  \cite{rocha-rewsmtjlamp-2017} we know that:
if $\phi_t \parallel t \rightsquigarrow^\ast \phi_u \parallel u$
then $t' \longrightarrow^\ast u'$
 for some $t' \in \llbracket \phi_t \parallel t \rrbracket$
 and $u' \in \llbracket \phi_u \parallel u \rrbracket$; and 
if  $t' \longrightarrow^\ast u'$ with
$t' \in \llbracket \phi_t \parallel t \rrbracket$,
then there exists $\phi_u$ and $t_u$ s.t. 
 $\phi_t \parallel t \rightsquigarrow^\ast \phi_u \parallel u$
with $u' \in \llbracket \phi_u \parallel u \rrbracket$.
\end{proof}

\noindent\textbf{Theorem \ref{th:folding}}

 \begin{proof}
 Let
 $\projnow{U}{} = \phi_u'\parallel t_u'$ and 
$\projnow{V}{}= \phi_v'\parallel t_v'$,
where $\ovars{t_u'} \cap \ovars{t_v'} = \emptyset$.
Let
$X_u = \ovars{\phi_u'} \setminus \ovars{t_u'}$
and
$X_v = \ovars{\phi_v'} \setminus \ovars{t_v'}$.
By construction,
$t_u', t_v' \in T_{\Sigma \setminus \Sigma_0}(X_0)$,
$\exists(\projnow{U}{}) = (\exists X_u) \phi_u'$,
and
$\exists(\projnow{V}{}) = (\exists X_v) \phi_v'$.
It suffices to show 
$\os \projnow{U}{} \cs \subseteq \os \projnow{V}{} \cs$
iff 
$U \preceq V$.

\noindent
 ($\Rightarrow$)
Assume $\os \projnow{U}{} \cs \subseteq \os \projnow{V}{} \cs$.
Then,
$t_u'$ and $t_v'$ are $E$-unifiable (witnessed by $w \in \os \projnow{U}{} \cs \cap \os \projnow{V}{} \cs$). Since 
$t_v'$
has no duplicate variables
and
$E$ only contains structural axioms for $\Sigma \setminus \Sigma_0$,
by the matching lemma \cite[Lemma~5]{rocha-rewsmtjlamp-2017},
there exists a substitution $\sigma$
with  $t_u' = t_v' \sigma$ (equality modulo ACU).
Since 
any built-in subterm of $t_u'$ is a variable in $X_0$,
$\sigma$ is a renaming substitution
$\sigma: X_0 \to X_0$
and 
thus $\os \projnow{V}{} \cs = \os (\projnow{V}{})\sigma \cs$.

Suppose $\exists(\projnow{U}{}) \Rightarrow \exists(\projnow{V}{})\sigma$ is not valid,
\ie{} $((\exists X_u) \phi_u') \wedge  (\forall X_v) \neg \phi_v'\sigma$ is satisfiable.
Let $Y$
be the set of free variables 
in $((\exists X_u) \phi_u') \wedge  (\forall X_v) \neg \phi_v'\sigma$.
Notice that $Y = \ovars{t_u'} = \ovars{t_v'\sigma}$.
Let $\rho : Y \to T_{\Sigma_0}$
be a ground substitution that represents 
a satisfying valuation of $(\exists X_u) \phi_u' \wedge  (\forall X_v) \neg \phi_v'\sigma$.
Then, $t_u'\rho \in \os \projnow{U}{} \cs$
but $t_u'\rho = t_v' \sigma\rho \notin \os (\projnow{V}{})\sigma \cs = \os \projnow{V}{} \cs$,
which is a contradiction.

\noindent ($\Leftarrow$) 
Assume $U \preceq V$.
There exists a substitution $\sigma: X_0 \to X_0$ such that $t_u' = t_v'\sigma$
and $\exists(\projnow{U}{}) \Rightarrow \exists(\projnow{V}{})\sigma$ is valid. 
Let $Y$ be the set of free variables in $\exists(\projnow{U}{}) \Rightarrow \exists(\projnow{V}{})\sigma$.
As mentioned above, 
$\os \projnow{V}{} \cs = \os (\projnow{V}{})\sigma \cs$
and
$Y = \ovars{t_u'} = \ovars{t_v'\sigma}$.
Let $w \in \os \projnow{U}{} \cs$.
Then,
for some ground substitution $\rho_u$, $w = t_u'\rho_u$
and
$\phi_u'\rho_u$ holds.
From the assignments in $\rho_u \vert_Y$, 
we can build a
valuation  $\mathcal{V}$ making true $\exists (\projnow{U}{})$
and, by assumption,  making also true $\exists (\projnow{V}{})\sigma$. 
Hence,
 there exists a ground substitution $\rho_v$ (that agrees on the values assigned
 in $\mathcal{V}$) such that $\phi_v'\sigma \rho_v$ holds 
 and $\rho_u \vert_Y = \rho_v \vert_Y$.
Notice that
$w = t_u'\rho_u = t_u' (\rho_u \vert_Y) = t_v'\sigma (\rho_v \vert_Y) = t_v'\sigma \rho_v$.
Therefore, $w \in \os (\projnow{V}{})\sigma \cs$.
\qed 
\end{proof}

\noindent\textbf{Corollary \ref{corollary:term}}
 \begin{proof}

     Assume that $(\mathcal{C}, (M,D), \Fleche{})$ is a finite transition
     system and, to obtain a contradiction, that there are infinitely many
     $\symredif$-reachable states from $\code{init}(\PN,M,D)$. Since 
     $\left(T_{\Sigma , \texttt{State}}, \code{init}(\PN,M,D), \symredif\right)$
     is finitely branching, there must be an infinite
     sequence of the form $U_0 ~\symredif~ U_1 ~\symredif~ \cdots $ where, by
     definition of $\symredif$, $U_j \not\preceq U_i$ for $i< j$. From 
     \Cref{th:folding} we know that $\os U_j\cs \not\subseteq \os U_i\cs$. By
     \Cref{th:sym-correct}, this means that after each transition, more concrete
     different states are found. Hence, the reachable state classes cannot be
     finite, thus a contradiction. 

 \end{proof}

 \end{document}